\documentclass[letterpaper, 10 pt, journal, twoside]{IEEEtran}
\usepackage[noadjust]{cite}

\usepackage{amsmath, amssymb,mathrsfs, dsfont} 

\usepackage{bm,bbm}
\DeclareMathAlphabet{\mathbbold}{U}{bbold}{m}{n}

\usepackage{xspace}
\newcommand{\MATLAB}{\textsc{Matlab}\xspace}
\usepackage{tikzsymbols}
\usepackage{color}
\usepackage{graphicx,graphics} 
\usepackage{epstopdf} 
\usepackage{caption}
\usepackage{algorithm, algorithmicx, multirow, titlecaps} %
\usepackage{threeparttable}
\usepackage[english]{babel}
\usepackage[noend]{algpseudocode} 
\usepackage{tabularx}
\usepackage{booktabs}
\usepackage{multirow}
\usepackage{courier}
\usepackage{float}
\usepackage{makecell}
\setcellgapes{5pt}

\usepackage{array}
\usepackage{adjustbox}

\usepackage{authblk}
\usepackage{url} 

\newcommand{\cmb}[1]{{\color{blue}{#1}}}

\newcommand{\mtiny}[1]{{\scalebox{.75}{#1}}}
\newcommand{\smtiny}[1]{{\scalebox{.63}{#1}}}
\newcommand{\stiny}[1]{{\scalebox{.5}{#1}}}

\newcommand{\argmin}{{\mathrm{argmin}}}

\newcommand*{\minOp}{\operatornamewithlimits{min}\limits}

\newcommand*{\sumOp}{\operatornamewithlimits{\sum}\limits}
\newcommand*{\limOp}{\operatornamewithlimits{lim}\limits}
\newcommand*{\limsupOp}{\operatornamewithlimits{limsup}\limits}
\newcommand*{\liminfOp}{\operatornamewithlimits{liminf}\limits}

\newcommand*{\capOp}{\operatornamewithlimits{\text{\scalebox{1.25}{$\cap$}}}\limits}

\newcommand{\tr}{{\smtiny{$\mathsf{T}$ }}\!}

\newcommand{\zero}{\mathbf{0}}

\iftrue
\newcommand{\zeromx}{{\mathbbold{0}}} 
\newcommand{\eye}{\mathbb{I}}
\else
\newcommand{\zeromx}{{\mathbf{0}}} 
\newcommand{\eye}{\mathbf{I}}
\fi
\newcommand{\rank}{\mathrm{rank}}
\newcommand{\vc}[1]{{ \mathrm{#1} }}
\newcommand{\mx}[1]{{ \mathrm{#1} }}
\newcommand{\diag}{\mathrm{diag}}

\newcommand{\inner}[2]{{ \langle {#1,#2} \rangle}}

\newcommand{\imag}{\mathrm{imag}}
\newcommand{\real}{\mathrm{real}}
\newcommand{\ol}[1]{\overline{#1}}

\newcommand{\ellone}{\ell^{1}}
\newcommand{\ellinfty}{\ell^{\infty}}
\newcommand{\ellp}{\ell^{p}}

\newcommand{\Dscr}{{\mathscr{D}}}

\newcommand{\Fscr}{{\mathscr{F}}}

\newcommand{\Pscr}{{\mathscr{P}}}

\newcommand{\Rscr}{{\mathscr{R}}}

\newcommand{\Tscr}{{\mathscr{T}}}

\newcommand{\Acal}{{\mathcal{A}}}

\newcommand{\Ccal}{{\mathcal{C}}}

\newcommand{\Ecal}{{\mathcal{E}}}

\newcommand{\Hcal}{{\mathcal{H}}}
\newcommand{\Ical}{{\mathcal{I}}}
\newcommand{\Jcal}{{\mathcal{J}}}

\newcommand{\Qcal}{{\mathcal{Q}}}

\newcommand{\Scal}{{\mathcal{S}}}

\newcommand{\Vcal}{{\mathcal{V}}}

\newcommand{\Xcal}{{\mathcal{X}}}

\newcommand{\Nbb}{{\mathbb{N}}}

\newcommand{\Rbb}{{\mathbb{R}}}

\newcommand{\Tbb}{{\mathbb{T}}}

\newcommand{\Zbb}{{\mathbb{Z}}}

\newcommand{\bbk}{\mathbbm{k}}

\newcommand{\RR}{\mathbb{R}}

\usepackage{amsthm}

\newtheorem{theorem}{Theorem}
\newtheorem*{theorem*}{Theorem}
\newtheorem{definition}{Definition}
\newtheorem*{definition*}{Definition}
\newtheorem{assumption}{Assumption}

\newtheorem{corollary}[theorem]{Corollary}

\newtheorem{remark}{Remark}

\newtheorem*{example*}{Example}

\newtheorem*{claim*}{Claim}

\newtheorem*{problem*}{Problem}

\iftrue 
\usepackage[colorlinks=true]{hyperref}
	\usepackage{xcolor}
	\hypersetup{
		colorlinks,
		linkcolor={blue!55!black},
		citecolor={red!55!black},
		urlcolor={blue!80!black}
	}
	\allowdisplaybreaks
	
\fi
\newcommand{\Hankel}{\mathrm{Hankel}}

\newcommand{\gS}{\vc{g}^{\stiny{$(\Scal)$}}}
\newcommand{\gtS}{{g}^{\stiny{$(\Scal)$}}}
\newcommand{\GS}{G^{\stiny{$(\Scal)$}} }
\newcommand{\geps}{\vc{g}^{\smtiny{($\varepsilon$)}}}
\newcommand{\gteps}{g^{\smtiny{($\varepsilon$)}}}
\newcommand{\Geps}{G^{\smtiny{($\varepsilon$)}}}
\newcommand{\FS}{F^{\stiny{$(\Scal)$}} }
\newcommand{\HS}{H^{\stiny{$(\Scal)$}} }
\newcommand{\NS}{N^{\stiny{$(\Scal)$}} }

\newcommand{\vca}{\vc{a}}
\newcommand{\vcb}{\vc{b}}
\newcommand{\vcc}{\vc{c}}

\newcommand{\vcf}{\vc{f}}
\newcommand{\vcg}{\vc{g}}
\newcommand{\vch}{\vc{h}}
\newcommand{\vcu}{\vc{u}}
\newcommand{\vcv}{\vc{v}}
\newcommand{\vcx}{\vc{x}}
\newcommand{\vcy}{\vc{y}}
\newcommand{\vcw}{\vc{w}}

\newcommand{\mxA}{\mx{A}}
\newcommand{\mxB}{\mx{B}}
\newcommand{\mxBr}{\mx{B}^{\supscrpsm{\mathrm{r}}}}
\newcommand{\mxBi}{\mx{B}^{\supscrpsm{\mathrm{i}}}}
\newcommand{\mxC}{\mx{C}}
\newcommand{\mxD}{\mx{D}}
\newcommand{\mxE}{\mx{E}}
\newcommand{\mxK}{\mx{K}}
\newcommand{\mxL}{\mx{L}}
\newcommand{\mxO}{\mx{O}}
\newcommand{\mxQ}{\mx{Q}}
\newcommand{\mxT}{\mx{T}}
\newcommand{\mxV}{\mx{V}}
\newcommand{\mxVr}{\mx{V}^{\supscrpsm{\mathrm{r}}}}
\newcommand{\mxVi}{\mx{V}^{\supscrpsm{\mathrm{i}}}}
\newcommand{\nD}{n_{\stiny{$\!\Dscr$}}}

\newcommand{\nG}{n_{\vc{g}}}

\newcommand{\kernel}{\bbk}

\newcommand{\Hk}{\Hcal_{\kernel}}
\newcommand{\Vk}{\Vcal_{\kernel}}
\newcommand{\TC}{\text{\mtiny{$\mathrm{TC}$}}}
\newcommand{\DC}{\text{\mtiny{$\mathrm{DC}$}}}
\renewcommand{\SS}{\text{\mtiny{$\mathrm{SS}$}}}
\newcommand{\kernelTC}{\kernel_{\TC}}
\newcommand{\kernelDC}{\kernel_{\DC}}
\newcommand{\kernelSS}{\kernel_{\SS}}
\newcommand{\Lu}[1]{\mx{L}^{\!\vc{u}}_{#1}}
\newcommand{\nx}{{n_{\mathrm{x}}}}
\newcommand{\Mx}{{m_{\mathrm{x}}}}

\newcommand{\mt}{{\smtiny{$\mathrm{(}m\mathrm{)}$}}}
\newcommand{\mbt}{{\smtiny{$\mathrm{(}\ol{m}\mathrm{)}$}}}

\newcommand{\supscrpsm}[1]{{\smtiny{$\mathrm{(}#1\mathrm{)}$}}}

\newcommand{\Pscra}{\Pscr_{\!{{\alpha}}}}
\newcommand{\Pscran}{\Pscr_{\!{{\alpha}},n}}
\newcommand{\Pscrna}{\Pscr_{\!{{\alpha}}}^{\supscrpsm{n}}}
\newcommand{\Fscrna}{\Fscr_{\!{{\alpha}}}^{\supscrpsm{n}}}
\newcommand{\Fscran}{\Fscr_{\!{{\alpha}},n}}
\newcommand{\PscrI}{\Pscr_{\!(0,1)}}
\newcommand{\PscrIn}{\Pscr_{\!(0,1),n}}
\newcommand{\PscrnI}{\Pscr_{\!(0,1)}^{\supscrpsm{n}}}

\newcommand{\expe}{\mathrm{e}}
\newcommand{\Jimage}{\mathrm{j}}

\newcommand{\ar}{a^{\supscrpsm{\mathrm{r}}}}
\newcommand{\ai}{a^{\supscrpsm{\mathrm{i}}}}
\newcommand{\vcar}{\vca^{\supscrpsm{\mathrm{r}}}}
\newcommand{\vcai}{\vca^{\supscrpsm{\mathrm{i}}}}
\newcommand{\vcfark}{\vcf_{{{\alpha}},k}^{\supscrpsm{\mathrm{r}}}}
\newcommand{\vcfaik}{\vcf_{{{\alpha}},k}^{\supscrpsm{\mathrm{i}}}}
\newcommand{\vcfaj}{\vcf_{{{\alpha}},j}}

\newcommand{\rhozero}{\rho_{\stiny{$\mathrm{d}$}}}
\usepackage{lineno}
\title{\LARGE Regularized Identification with Internal Positivity Side-Information}
\author[$\dagger$ ]{Mohammad Khosravi}%
\author[$\dagger$ ]{Roy S. Smith}
\affil[$\dagger$]{Automatic Control Laboratory, ETH Z\"urich 
\authorcr
\texttt{\{khosravm,rsmith\}@control.ee.ethz.ch}
}
\begin{document}
\bstctlcite{IEEEexample:BSTcontrol}
\maketitle
\begin{abstract}
In this paper, we present an impulse response identification scheme that incorporates the internal positivity side-information of the system. The realization theory of positive systems establishes specific criteria for the existence of a positive realization for a given transfer function. These transfer function criteria are translated to a set of suitable conditions on the shape and structure of the impulse responses of positive systems. Utilizing these conditions, the impulse response estimation problem is formulated as a constrained optimization in a reproducing kernel Hilbert space equipped with a stable kernel, and suitable constraints are imposed to encode the internal positivity side-information. The optimization problem is infinite-dimensional with an infinite number of constraints. An equivalent finite-dimensional convex optimization in the form of a convex quadratic program is derived. The resulting equivalent reformulation makes the proposed approach suitable for numerical simulation and practical implementation. A Monte Carlo numerical experiment 
evaluates the impact of incorporating the internal positivity side-information in the proposed identification scheme. The effectiveness of the proposed method is demonstrated using data from a heating system experiment.	
\end{abstract}

\section{Introduction}\label{sec:int}
In various dynamical systems, the characteristic variables are constrained to be non-negative or bounded, either by the nature of their definition or according to the physics of the underlying system. For instance, charges in RC-circuits, temperatures and thermal energies in buildings, mass flows in compartmental systems, the population of certain species of animals or bacteria, the concentration of pathogens, level of traffic and congestion in networks and roads, prices of stocks and goods, the pressure of fluids, and many other quantities of interest are always non-negative \cite{haddad2010nonnegative, brown1980compartmental,shorten2006positive,krause2015positive}.
In the broad sense, a system described only by such non-negative variables is called a \emph{positive system} \cite{farina2011positive}.
Depending on our perspective, i.e., whether the positivity feature is considered as an input-output property or a state-space characteristic, we have two central notions of positivity in the system theory literature \cite{farina2011positive}: internal positivity and external positivity, where the main focus of the literature is on the former one. In externally positive systems, the non-negativity of the input signal implies the same feature for the output signal. Meanwhile, in internally positive systems, the state trajectory and output signal are non-negative when the initial state and input signal are non-negative.

Positive systems have received extensive attention in the past decades owing to being omnipresent in various fields of science and their wide range of applications \cite{rantzer2018tutorial,benvenuti2004tutorial}. Luenberger pioneered the system theoretic approach to positive systems with his seminal work \cite{luenberger1979introduction} in the 1980s. Since then various subjects of system theory aspects have been tackled, e.g., realization theory \cite{benvenuti2004tutorial}, controllability and reachability \cite{valcher2009reachability}, observability and observer design \cite{back2008design}, robust stability  \cite{hinrichsen2007robust, colombino2015convex}, positive stabilization \cite{roszak2009necessary, shafai1997explicit}, fault detection and estimation \cite{oghbaee2018complete}, decentralized and distributed control \cite{dhingra2018structured, ebihara2012decentralized}, and, large scale positive systems and scalable control \cite{ebihara2013stability,  rantzer2021scalable}.

Concerning the identification problem, there are two aspects when the underlying system is known to be internally positive. First, respecting the positivity property can be an essential issue in some applications, such as implementing model predictive control. Accordingly, any accurate mathematical modeling approach is expected to include this feature and construct an internally positive system. The second aspect is informed system identification \cite{ahmadi2020learning,khosravi2019positive,khosravi2021ROA,khosravi2021grad} and concerns utilizing the internal positivity side-information and integrating features of this knowledge in the model to improve the estimation accuracy. Indeed, disregarding the positivity information can lead to models, which are not physically interpretable and explainable, or behaviors in contradiction with our expectations \cite{umenberger2016scalable,rantzer2018tutorial}. While positive systems have been extensively researched from various viewpoints \cite{rantzer2021scalable}, their identification problem is not well studied, especially with regard to these aspects. For instance, a set of conditions is introduced in \cite{benvenuti2002model} for identifying compartmental models, which is a particular case of internal positivity feature. In \cite{de2002identification}, assuming the output sequence of data is a Poisson process, a maximum likelihood approach is presented for third-order positive systems with distinct real poles. In \cite{umenberger2016scalable}, a particular situation is considered, where state variable measurements are provided, in addition to input-output data, and the stability and scalability issues are discussed. Since internally positive systems are also externally positive, side-information on internal positivity implies external positivity. Accordingly, from the perspective of informed system identification, one can consider external positivity as partial information to be integrated into the model. To this end, one may employ the external positive system identification methods \cite{grussler2017identification,zheng2021bayesian}. For example, since the externally positive systems are precisely  those with the  non-negative impulse response, a kernel-based nonparametric maximum a posteriori approach is introduced in \cite{zheng2018positive,zheng2021bayesian} for estimating a non-negative finite impulse response (FIR). In this approach, the covariance of the prior distribution is specified by the stable kernels, while the mean is designed arbitrarily as an exponentially decaying FIR. Note that in the FIR identification approaches with an external positivity constraint \cite{grussler2017identification,zheng2018positive,zheng2021bayesian}, the complete information of internal positivity is not exploited. The kernel-based methods \cite{pillonetto2010new}, which resolve the issues of bias-variance trade-off, robustness, and model order selection \cite{ljung2020shift,khosravi2021robust,pillonetto2014kernel,khosravi2020regularized}, also provide a suitable framework for the integration of various sorts of side-information into the model   \cite{fujimoto2017extension, chen2012estimation,khosravi2020low, darwish2018quest, chen2018kernel, marconato2016filter, khosravi2021SSG, risuleo2017nonparametric, risuleo2019bayesian, everitt2018empirical, khosravi2021FDI}. These include the stability of the system, the smoothness of the impulse response, time constants and resonant frequencies \cite{chen2018kernel}. Together with the realization theory of positive systems \cite{farina2011positive}, the kernel-based framework can provide a suitable foundation for impulse response identification of positive systems.

This paper extends our previous work \cite{khosravi2019positive} and presents an identification method that integrates the internal positivity side-information in the estimated impulse response. From the realization theory of positive systems \cite{farina2011positive}, we know that the impulse response of an internally positive system has a specific form, i.e., it has a dominant non-negative part, where the corresponding transfer function has structured poles, and a residual part. This specific form can be translated to a set of structural constraints on the impulse response. Accordingly, the estimation problem is expressed in the form of a constrained optimization in a stable reproducing kernel Hilbert space, where suitable constraints are imposed to encode the internal positivity side-information. Though this problem is initially formulated in an infinite-dimensional space and with an infinite number of constraints, we derive an equivalent finite-dimensional convex optimization in the form of a convex quadratic program. We evaluate the impact of incorporating the internal positivity side-information and assess the performance of the proposed identification scheme through a Monte Carlo numerical experiment. The efficacy of the proposed positive system identification technique is confirmed using data from a thermal dynamics experiment.

\section{Notations}\label{sec:notations}
In this paper, the set of natural numbers, the set of integers, the set of non-negative integers, the set of real numbers, the set of non-negative real numbers, $n$-dimensional Euclidean space, and the set of $n$ by $m$ matrices are respectively shown by $\Nbb$, $\Zbb$, $\Zbb_+$, $\Rbb$, $\Rbb_+$, $\Rbb^n$, and $\Rbb^{n\times m}$.
The positive orthant of $\Rbb^n$ is denoted by $\Rbb_+^n$.
The identity matrix and zero matrix are denoted by  $\eye$ and  $\zeromx_{n}$, respectively. 
Also, the $n$-dimensional zero vector and the all-ones vector are denoted by $0_n$ and $1_n$, respectively. When the dimension is clear from the context, we drop the subscript. 
For $p\in[1,\infty)$, the $p$-norm of vector $\vch = (h_s)_{s=0}^{\infty}\in\Rbb^{\Zbb_+}$  is defined as $\|\vch\|_p = (\sum_{s\ge 0}|h_s|^p)^{\frac{1}{p}}$,
and the $\infty$-norm of $\vch$ is defined as $\|\vch\|_{\infty} = \sup_{s\ge 0}|h_s|$. The space of vectors $\vch$ with finite $p$-norm is denoted by $\ellp$.
With respect to each $\vcu =(u_s)_{s=0}^{\infty}\in\ellinfty$ and $t\in\Zbb$, the  linear map $\Lu{t}:\ellone\to\Rbb$ is defined as  $\Lu{t}(\vcg) = \sum_{s=0}^{\infty}g_s u_{t-s}$, for any $\vcg=(g_s)_{s=0}^{\infty} \in\ellone$.
Given a subset $\Ccal\subset\Xcal$, the function $\delta_{\Ccal}:\Xcal\to\{0,+\infty\}$  is defined as $\delta_{\Ccal}(x) = 0$, if $x\in\Ccal$ and  $\delta_{\Ccal}(x) = \infty$, otherwise.
The set of polynomials in $x$ with  maximum degree $n$ and real coefficients is denoted by $\Rbb_n[x]$.
For transfer function $G$, $r(G)$ denotes its spectral radius.

\section{System Identification with Internal Positivity Side-Information}
\label{sec:pf}
Let $\gS := (\gtS_t)_{t=0}^\infty$ be the impulse response of stable and causal system $\Scal$ and $\GS(z):=\sum_{t=0}^\infty \gtS_tz^{-t}$ be the corresponding transfer function.
We call impulse response $\gS$, or equivalently the system $\Scal$, \emph{internally positive} if there exists a realization such that the state trajectory and the output remains non-negative given that the initial state and the input are non-negative (see Definition \ref{def:in_pos}). 
Suppose a bounded signal $\vcu\in\ellinfty$ is applied to the input of system $\Scal$. Let $y_t$ denote the measured output at time instant $t\in\Tscr$, where $\Tscr:=\{t_i\ \!|\! \ i=0,\ldots,\nD\!-\!1\}$, for a given $\nD\in\Nbb$.
In other words, we have 
\begin{equation}\label{eqn:output_sys_S}
	y_t := \Lu{t}(\gS)+w_t, \qquad t\in\Tscr,
\end{equation}
where $w_t$  denotes the uncertainty in the output measured at time instant $t$, for $t\in\Tscr$.
Accordingly, we have a set of input-output measurement data denoted be $\Dscr$.
Based on the introduced setting, we introduce the following impulse response identification problem.
\begin{problem*}
Using data $\Dscr$, estimate the impulse response of $\gS$, given the side-information that $\gS$ is internally positive.	
\end{problem*}

In addressing this problem, the main concern is the appropriate integration of the available internal positivity side-information into the impulse response identification problem. To this end, we need to exploit suitable conditions inducing the desired positivity feature. In the next section, these conditions are discussed, and the estimation problem is formulated accordingly.

\section{Mathematical Formulation for Internally Positive System Identification}
\label{sec:positive_realization}
In the realization theory of positive systems, 
sufficient conditions are introduced under which the transfer function of a system admits a so-called positive realization.
We employ these conditions together with the notion of stable reproducing kernel Hilbert spaces (RKHS) for bridging to the impulse response identification of positive systems. 

The following definition introduces the external positivity  notion for impulse response $\gS$, or equivalently for system $\Scal$.

\begin{definition}[\cite{farina2011positive}]\label{def:ex_pos}
	The impulse response $\gS$ is said to be {\em externally positive} if for each  $\tau\in\Zbb$ and for any input signal $\vcu=(u_t)_{t\in\Zbb}$, 
	we know that the (noiseless) output of the system 
	is non-negative  for all $t\le \tau$, when we have $u_t\ge 0$, for each $t\le \tau$.
	The set of externally positive impulse responses
	is denoted by $\overline{\Pscr}$.
\end{definition}

Since the condition introduced in Definition \ref{def:ex_pos} is imposed for all $\tau\in\Zbb$, one can see that externally positivity is equivalent to the non-negativity of input signals implying the non-negativity of the output signal of the system, 
i.e., if $u_t\ge 0$, for all $t\in\Zbb$, then  $y_t=\Lu{t}(\gS)\ge 0$, for each $t\in\Zbb$.
The next theorem introduces a necessary and sufficient condition for external positivity.
\begin{theorem}[\cite{farina2011positive}]\label{thm:ex_pos}
	The impulse response $\gS$ is externally positive if and only if
	$g_t\ge 0$, for any $t\in\Zbb_+$.	
\end{theorem}
From Theorem \ref{thm:ex_pos}, one can see that the external positive stable impulse responses are exactly the non-negative ones, i.e., we have
$\overline{\Pscr} = \ellone \cap \Rbb_+^{\Zbb_+}$. Along with the external positivity notion, we have another notion of positivity 
which is the main concern of the current paper and introduced in the following definition.
 
\begin{definition}[\cite{farina2011positive}]\label{def:in_pos}
The impulse response $\gS$ is said to be {\em internally positive}, or simply {\em positive}, 
if there exists a realization for system $\Scal$ as 
\begin{equation}\label{eqn:S_Abcd}
	\left\{
	\begin{array}{ccl}
		\vcx_{t+1} 
		\!\!&\!\!\!=\!\!\!&\!\! 
		\mxA\vcx_t + \vcb u_t,\\
		y_{t}   
		\!\!&\!\!\!=\!\!\!&\!\! 
		\vcc\vcx_t + d u_t,\\
	\end{array}
	\right.
	\qquad \forall\ t\in\Zbb,
\end{equation}
where $\mxA\in\Rbb^{\nx\times \nx}$, $\vcb\in\Rbb^{\nx}$, $\vcc\in\Rbb^{1\times \nx}$, $d\in\Rbb$ and $\nx\in\Nbb$, 
such that $\vcx_0\in\Rbb_+^{\nx}$ and $u_t\in\Rbb_+$, for all $t \ge 0$, implies that
$\vcx_t\in\Rbb_+^{\nx}$ and $y_t\in\Rbb_+$, for each $t\ge 0$. 
The realization \eqref{eqn:S_Abcd} with the this property is called a \emph{positive realization} for $\gS$, or equivalently, for $\GS$.
Moreover, the set of internally positive impulse responses is denoted by $\Pscr$.
\end{definition}
The internal positivity enforces a specific attribute on $\gS$ according to the \emph{Kronecker's theorem}  given below.
\begin{theorem}[\cite{fuhrmann2011polynomial}]\label{thm:kronecker}
With respect to impulse response $\vcg=(g_t)_{t=0}^{\infty}$, define Hankel operator $\Hankel(\vcg):\ellinfty
\to \ellinfty$ with entrywise representation in the standard basis of $\ellinfty$ as following
\begin{equation}
	\Hankel(\vcg) :=
	\begin{bmatrix}
		g_0 & g_1 & g_2 & \ldots\\
		g_1 & g_2 & \reflectbox{$\ddots$} & \\
		g_2 & \reflectbox{$\ddots$} &  & \\
		\vdots &  &  &\\		
	\end{bmatrix}
	= \begin{bmatrix}g_{i+j-2}\end{bmatrix}_{i,j=1}^{\infty}.
\end{equation}
Then, $G(z):=\sum_{t=0}^{\infty}g_tz^{-t}$ is a rational function if and only if the rank of $\Hankel(\vcg)$ is finite, i.e., we have 
\begin{equation*}
\rank\big(\Hankel(\vcg)\big)=
\dim\Big\{\Hankel(\vcg)\vcv\ \! \big| \ \!  \vcv\in\ellinfty\Big\}<\infty.
\end{equation*}
We call impulse response $\vcg$ \emph{finite Hankel rank} when  the property  above is satisfied.
\end{theorem}
According to \eqref{eqn:S_Abcd}, we know that $\GS(z)=\vcc(z\eye-\mxA)^{-1}\vcb+d$ is a rational function.
Therefore, due to Theorem \ref{thm:kronecker}, the internal positivity of $\gS$ implies that
\begin{equation}\label{eqn:rank_Hankel_gS}
\rank\big(\Hankel(\gS)\big)<\infty.	
\end{equation}

The realization \eqref{eqn:S_Abcd} needs to have special structure introduced in the
next theorem.

\begin{theorem}[\cite{farina2011positive}]\label{thm:in_pos}
The impulse response $\gS$ is internally positive if and only if there exist a realization as in \eqref{eqn:S_Abcd} such that the entries of $\mxA,\vcb,\vcc$ and $d$ are non-negative. 
\end{theorem}
The next corollary is easily concluded from Theorem \ref{thm:in_pos} for the internally positive impulse responses.  
\begin{theorem}[\cite{farina2011positive}]\label{thm:in_pos_ex_pos}
If impulse response $\gS$ is internally positive then we have 
\begin{equation}\label{eqn:gS_t_ge_0}
	\gtS_t\ge 0, \qquad \forall t\in\Zbb_+.
\end{equation}
\end{theorem}
Theorem \ref{thm:in_pos_ex_pos} implies that any internally positive impulse response is externally positive as well. Accordingly, due to \eqref{eqn:rank_Hankel_gS}, we have $\Pscr\subset\underline{\Pscr}$, where the set of impulse responses $\underline{\Pscr}$ is defined as
\begin{equation}\label{eqn:P_subset_barP_finite_order}
\underline{\Pscr}:=\overline{\Pscr}\cap
\Big\{\vcg\!\in\!\ellone\!\ \big|\!\ \rank(\Hankel(\vcg))<\infty\Big\}.	
\end{equation}
Meanwhile, from the next example, which is a modified version of an example given in \cite{benvenuti2004tutorial}, one can see that the inclusion in $\Pscr\subset\underline{\Pscr}$ is strict, i.e., $\Pscr\neq\underline{\Pscr}$.

\begin{example*}
\normalfont
Let ${{\rho}}\in (0,1)$ and
$\omega$ be an irrational real number.
Define the impulse response $\vcg=(g_t)_{t=0}^{\infty}$ as
\begin{equation}\label{eqn:g_example}
	g_t = {{\rho}}^t(1+\cos(2\pi\omega t)),\qquad\forall t\in\Zbb_+.
\end{equation}
One can see that $\vcg$ is a non-negative impulse response with following transfer function 
\begin{equation}
	G(z) = \frac{1}{1-{{\rho}} z^{-1}}+
	\frac{1-{{\rho}}\cos w \ z^{-1}}
	{1-2\rho \cos w \ z^{-1}+\rho^2  z^{-2}}.
\end{equation}
Therefore, we have $\vcg\in\underline{\Pscr}$. However, there is no positive realization for the impulse response $\vcg$ \cite{farina2011positive}. Accordingly, due to Definition \ref{def:in_pos},  we know that $\vcg$ is not internally positive, i.e., $\vcg\notin\Pscr$ and $\Pscr\neq\underline{\Pscr}$.
\end{example*}

Based on the above discussion, conditions \eqref{eqn:rank_Hankel_gS} and \eqref{eqn:gS_t_ge_0} are necessary but not sufficient for $\gS$ being an internally positive impulse response.
Using the following theorem, we derive sufficient conditions for the internal positivity of $\vcg$ which are used later to formulate the identification problem of internally positive systems.
\begin{theorem}[\cite{farina2011positive}]
\label{thm:pos_real}
Let $\vcg$ be a non-negative impulse response and $G$ be the corresponding transfer function
If $G$ is a strictly proper rational function with a unique dominant pole $\rho\in(0,1)$, then there exists a positive realization for $G$ .
\end{theorem}
With respect to each $\rho\in(0,1)$, define $\Pscr_{\!\rho}\subset\ellone$ as the set of non-negative impulse responses satisfying \eqref{eqn:rank_Hankel_gS} 
such that we have
\begin{equation}
	\label{eqn:Exist_lim_alpha_tg_t}
	\exists\ \! {a}\!\in\!(0,\infty), \quad 
	\lim_{t\to \infty} \rho^{-t}g_t = {a},
\end{equation}
i.e., $\lim_{t\to \infty} \rho^{-t}g_t$
is well-defined and equal to a positive real scalar ${a}$. 
Furthermore, we define $\Pscr_{(0,1)}$ as $\Pscr_{\!(0,1)} = \cup_{\rho\in(0,1)} \Pscr_{\!\rho}$.
Based on Theorem \ref{thm:pos_real}, 
we have the following corollary for $\Pscr_{\!\rho}$ and $\Pscr_{\!(0,1)}$.
\begin{corollary}
	\label{cor:dom_pole_simple}
	For any $\rho\in(0,1)$, each impulse response $\gS\in\Pscr_{\!\rho}$ is internally positive, i.e., $\Pscr_{\!\rho}\subset\Pscr$.
	Moreover, we have $\Pscr_{\!(0,1)}\subset\Pscr$.
\end{corollary}
\begin{proof}
See Appendix \ref{sec:pf_dom_pole_simple}.
\end{proof}
Note that $\Pscr_{\!(0,1)}$ contains exactly the impulse responses satisfying conditions \eqref{eqn:rank_Hankel_gS}, \eqref{eqn:gS_t_ge_0},  and, \eqref{eqn:Exist_lim_alpha_tg_t}.
Hence, Corollary \ref{cor:dom_pole_simple} says that any impulse response in $\ellone$ which satisfies these conditions is internally positive.  
Accordingly, one can employ \eqref{eqn:rank_Hankel_gS}, \eqref{eqn:gS_t_ge_0}  and \eqref{eqn:Exist_lim_alpha_tg_t} in the identification problem to enforce internal positivity on the  impulse response to be estimated. 
The next theorem further highlights the importance of positive systems $\Pscr_{\!(0,1)}$.
\begin{theorem}\label{thm:denseness_IP_EP}
	The set of impulse responses $\Pscr_{\!(0,1)}$ is dense in $\Pscr$ with respect to $p$-norm topology, for any $p\in[1,\infty]$.
\end{theorem}
\begin{proof}
See Appendix \ref{sec:pf_denseness_IP_EP}.
\end{proof}
With respect to each $\vcg=(g_t)_{t=0}^\infty\in \Pscr_{(0,1)}$,
one can define impulse response $\vch=(h_t)_{t=0}^\infty$
such that  
$h_t=g_t-{a}\rho^t$, for $t\in\Zbb_+$, where $\rho$ and ${a}$ are the positive scalars introduced in \eqref{eqn:Exist_lim_alpha_tg_t}.
Note that  
$\lim_{t\to \infty} \rho^{-t}h_t = 0$ and $\rho\in(0,1)$, which implies that $\vch = (h_t)_{t=0}^{\infty}$ is a stable impulse response dominated 
by 
\begin{equation}
\vcf_{\rho}=(f_t)_{t=0}^\infty:=(\rho^t)_{t=0}^\infty.
\end{equation}
Since, for each $t$, we have $g_t={a}\rho^t+h_t$, to identify the internally positive impulse response $\vcg=(g_t)_{t=0}^\infty$, we need to estimate $\rho,{a}$ and the stable impulse response $\vch=(h_t)_{t=0}^\infty$, dominated by $\vcf_{\rho}=(\rho^t)_{t=0}^\infty$, and meanwhile ensure that $\vcg=(g_t)_{t=0}^\infty$ satisfies properties \eqref{eqn:rank_Hankel_gS} and \eqref{eqn:gS_t_ge_0}.
Accordingly, we need a suitable hypothesis space for $\vch=(h_t)_{t=0}^\infty$. 
To this end, we employ stable reproducing kernel Hilbert spaces introduced below.
\begin{definition}[\cite{berlinet2011reproducing,chen2018stability}]
	\label{def:kernel_and_section}
	The non-zero symmetric function $\kernel:\Zbb_+\times\Zbb_+\to \Rbb$ is said to be a \emph{Mercer kernel} if, 
	for any $m\in\Nbb$, $t_1,\ldots,t_m\in\Tbb$ and $a_1,\ldots,a_m\in\Rbb$,
	we have 
	$\sum_{i=1}^{m}\!\sum_{j=1}^{m}\! a_i\kernel(t_i,t_j)a_j\ge 0$.
	Moreover, the \emph{section} of kernel $\kernel$ at $t\in\Zbb_+$ is denoted by $\kernel_{t}$ and
	defined as the function $\kernel(t,\cdot):\Zbb_+\to\Rbb$.
	Furthermore, the positive kernel $\kernel$ is said to be \emph{stable} if, for any $\vcu=(u_t)_{t\in\Zbb_+}\in\ellinfty$, we have
	$\sum_{t\in\Zbb_+}\!\!|	\sum_{s\in\Zbb_+}\!\! u_s\kernel(t,s)|<\infty$.	
\end{definition}
\begin{theorem}[\cite{berlinet2011reproducing,chen2018stability}]\label{thm:kernel_to_RKHS_def}
	Given a Mercer kernel $\kernel:\Zbb_+\times\Zbb_+\to \Rbb$, there exists a \emph{unique} Hilbert space $\Hk\subseteq \Rbb^{\Zbb_+}$ endowed with inner product $\inner{\cdot}{\cdot}_{\Hk}$ and norm $\|\cdot\|_{\Hk}$, called a \emph{RKHS with kernel} $\kernel$, such that, for each $t\in\Zbb_+$, we have
	\begin{itemize}
		\item[i)] $ \kernel_t\in\Hk$, and
		\item[ii)] $\inner{\vcg}{ \kernel_{t}}_{\Hk}=g_t$, for all $\vcg=(g_t)_{t\in\Zbb_+}\in\Hk$.
	\end{itemize} 
	The second feature is called {\em reproducing property}.
	Moreover, $\Hk\subset\ellone$ if and only if $\kernel$ is a stable kernel. In this case, $\Hk$ is said to be a \emph{stable RKHS}.
\end{theorem}
Given a stable kernel $\kernel$, we take $\Hk$ as the hypothesis space for the stable impulse response $\vch$.
Considering the set of input-output data $\Dscr$, we define the empirical loss function
$\Ecal_{\rho}:\Rbb\times\Hk\to\Rbb_+$ as 
\begin{equation}
\Ecal_{\rho}({a},\vch) := \sum_{i=0}^{\nD-1}
\Big(
y_{t_i}-{a}\Lu{t_i}(\vcf_{\rho})
-\Lu{t_i}(\vch)
\Big)^2,
\end{equation}
where we assume the hyperparameter $\rho\in(0,1)$ is given.
The estimation of $\rho$ is will be discussed later.
We formulate the identification problem with internal positivity side-information as following regularized optimization 
\begin{equation}\label{eqn:opt_h_rank}
\begin{array}{cl}
	\minOp_{{a}\in\Rbb,\vch\in\Hk} & 
	\Ecal_{\rho}({a},\vch) + \lambda\!\ \|\vch\|_{\Hk}^2,\\
	\text{s.t.}
	&h_t+{a}\rho^t\ge 0, \quad \forall t \ge 0,\\
	&\rank(\Hankel(\vch))< \infty,\\
	&{a}\ge {a}_{\min},
\end{array} 
\end{equation}
where ${a}_{\min}>0$ is a given lower-bound for ${a}$ to ensure that ${a}>0$, and $\lambda>0$ is the regularization weight. 
Note that, similarly to the standard problem formulation in the literature on the kernel-based impulse response identification \cite{pillonetto2014kernel},  the objective function in \eqref{eqn:opt_h_rank} is an empirical loss function regularized with the RKHS norm of $\vch$. 
This ensures the stability of $\vch$ and also allows incorporating other features such as exponential decay and smoothness \cite{chen2018kernel}.
The next theorem says that the solution of \eqref{eqn:opt_h_rank}  leads to an internally positive estimation of impulse response $\vcg$.
Before proceeding to the theorem, we need to introduce an assumption.
\begin{assumption}
	\label{ass:kernel_dominated}
	There exist $C\in\Rbb_+$ and $\rhozero\in(0,\rho)$ such that we have $|\kernel(t,t)|\le C \rhozero^{2t}$, for any $t\in\Zbb_+$.
\end{assumption}
\begin{theorem}\label{thm:opt_h_rank}
Let Assumption \ref{ass:kernel_dominated} hold, ${a}$ and $\vch = (h_t)_{t=1}^\infty$ be a solution pair for \eqref{eqn:opt_h_rank}, and, the impulse response $\vcg=(g_t)_{t=0}^\infty$ be defined as $g_t = h_t +{a} \rho^t$, for any $t\in\Zbb_+$.
Then, $\vcg$ is internally positive.  
\end{theorem}
\begin{proof}
	Let $G$ be the transfer function which corresponds to $\vcg$.
	Due to \eqref{eqn:opt_h_rank}, the rank of Hankel operator $\Hankel(\vch)$ is finite. Subsequently, according to Theorem~\ref{thm:kronecker}, the  transfer function corresponding to $\vch$, $H$, has finite order. 
	On the other hand, we know that 
	\begin{equation}\label{eqn:Gz}
		G(z)=\frac{{a} z^{-1}}{1-\rho z^{-1}}+ H(z).
	\end{equation}
	Therefor, the order of $G$ is finite. Accordingly, due to Theorem~\ref{thm:kronecker}, we know  that 
	$\rank(\Hankel(\vcg))<\infty$, i.e., $\vcg$ satisfies \eqref{eqn:rank_Hankel_gS}.
	Furthermore, according to the first constraint  in \eqref{eqn:opt_h_rank}, one can see that $\vcg$ is a non-negative impulse response and \eqref{eqn:gS_t_ge_0} holds for $\vcg$.
	Moreover, from the reproducing property of kernel,  we know that $h_t = \inner{\kernel_t}{\vch}$ and $\|\kernel_t\|_{\Hk}^2 = \inner{\kernel_t}{\kernel_t} = \kernel(t,t)$, for any $t\in\Zbb_+$. Consequently, due to Cauchy-Schwartz inequality and Assumption \ref{ass:kernel_dominated}, we have 
	\begin{equation}
	\begin{split}\!\!\!\!\!
	|h_t| = |\inner{\kernel_t}{\vch}|
	&\le \|\kernel_t\|_{\Hk}\|\vch\|_{\Hk}
	\\&\qquad\ = \kernel(t,t)^{\frac12}\|\vch\|_{\Hk}
	\le C^{\frac{1}{2}}\|\vch\|_{\Hk}\rhozero^t.
	\end{split}	
	\end{equation}
	for any $t\in\Zbb_+$.
	Following this, one can see that
	\begin{equation}
	\begin{split}
	0&
	\le 
	\liminfOp_{t\to \infty}\rho^{-t}h_t 
	\le 
	\limsupOp_{t\to \infty}\rho^{-t}h_t 
	\le
	\limsupOp_{t\to \infty}\rho^{-t}|h_t|
	\\&
	\le 
	\limsupOp_{t\to\infty}
	C^{\frac{1}{2}}\|\vch\|_{\Hk}\rhozero^t\rho^{-t} = 
	0, 
	\end{split}
	\end{equation}
	where the last equality is due to $\rhozero\in(0,\rho)$.  
	Hence, $\lim_{t\to\infty}\rho^{-t}h_t$ is well-defined and we have $\lim_{t\to\infty}\rho^{-t}h_t=0$.
	Subsequently, due to the definition of $\vcg$, it follows that  
	$\lim_{t\to\infty}\rho^{-t}g_t={a}$ and $\vcg$ satisfies \eqref{eqn:Exist_lim_alpha_tg_t}.
	Therefore, $\vcg$ belongs to $\Pscr_{\!\rho}$, and consequently, due to Corollary \ref{cor:dom_pole_simple}, $\vcg$ is internally positive.
\end{proof}
\begin{remark}\normalfont
	For $a$, $\vch$ and $\vcg$ introduced in Theorem~\ref{thm:opt_h_rank}, we have 
	\begin{equation}
		\Ecal_{\rho}(a,\vch)
		=
		\sum_{i=0}^{\nD-1}
		\big(
		y_{t_i}-\Lu{t_i}(\vcg)
		\big)^2,
	\end{equation}
	i.e., in the cost function of \eqref{eqn:opt_h_rank}, the first term  is the sum of squared errors for the impulse response fitting when the dominant pole $\rho$ is known. 
\end{remark}
\section{Towards a Tractable Solution}\label{sec:towards_tractable}
\label{sec:Towards_a_Tractable_Solution}
In this section, we investigate optimization problem \eqref{eqn:opt_h_rank} which was introduced for impulse response identification with internal positivity.
This optimization problem is in an infinite-dimensional space with an infinite number of constraints.
In the following, we analyze this problem and provide a tractable approach for deriving its solution.

Let $\Vk$ be the Hilbert space $\Rbb\times\Hk$ which is endowed with inner product  $\inner{\cdot}{\cdot}_{\Vk}:\Vk\times\Vk\to\Rbb$ defined as following
\begin{equation}
\inner{({a}_1,\vch_1)}{({a}_2,\vch_2)}_{\Vk} = {a}_1{a}_2 + \inner{\vch_1}{\vch_2}_{\Hk}, 	
\end{equation}
for any ${a}_1,{a}_2\in\Rbb$ and $\vch_1,\vch_2\in\Hk$.
Also, let $\Fscr\subseteq \Hk$ be the set of finite Hankel rank impulse responses in $\Hk$, i.e., 
\begin{equation}
\Fscr =  \big\{\vch\in \Hk \!\ \big| \!\ \rank( \Hankel(\vch)) < \infty \big\}.
\end{equation}
We define function $\Jcal_{\!\Fscr}:\Vk\to \RR\cup\{+\infty\}$ as  
\begin{equation}\label{eqn:J_F}
	\Jcal_{\!\Fscr}({a},\vch) =  \Ecal_{\rho}({a},\vch)
	+
	\sum_{s=0}^{\infty}\delta_{{\Rscr}_s}({a},\vch)
	+
	\delta_{\Fscr}(\vch)
	+
	\lambda \|\vch\|_{\Hcal_{\bbk}}^2,
\end{equation}
where ${\Rscr}_s\subseteq \Vk$ is the following set 
\begin{equation}
{\Rscr}_s:= 
\Big\{
\big({a},(h_s)_{s\in\Zbb_+}\!\big)\in\Vk \!\ \Big| \!\  h_s+{a}\rho^s\ge 0, {a}\ge {a}_{\min}
\Big\},	
\end{equation}
for $s\in\Zbb_+$.
From the definition of $\Jcal_{\!\Fscr}$,  it follows easily that the
optimization problem \eqref{eqn:opt_h_rank} is equivalent to 
\begin{equation}
	\label{eqn:opt_h_rank_noConstraint}
	\inf_{({a},\vch)\in\Vk}\ \! \Jcal_{\!\Fscr}({a},\vch).
\end{equation}
For $({a},\vch)=({a}_{\min},\zero)$, where $\zero$ denotes the zero vector in $\Hk$, one can easily see that 
\begin{equation}
	\Jcal_{\!\Fscr}({a}_{\min},\zero)
	= 
	\sum_{i=1}^n
	\big(y_{t_i}-
	{a}_{\min}\Lu{t_i}(\vcf_{\rho})\big)^2 <\infty.
\end{equation}
Since, for any $({a},\vch)\in\Vk$, we have $\Jcal_{\!\Fscr}({a},\vch)\ge 0$,  it follows that 
\eqref{eqn:opt_h_rank_noConstraint} is bounded. 
However, this arguments does not guarantee the existence of solution for  
\eqref{eqn:opt_h_rank_noConstraint}.
In the following, we show that under mild conditions the optimization problem \eqref{eqn:opt_h_rank_noConstraint} admits a solution when the kernel $\kernel$ meets certain criteria.
Let function
$\Jcal:\Vk\to \RR\cup\{+\infty\}$ be defined as following  
\begin{equation}\label{eqn:J}
	\Jcal({a},\vch)  =  \Ecal_{\rho}({a},\vch)
	+
	\sum_{s=0}^{\infty}\delta_{{\Rscr}_s}({a},\vch)
	+
	\lambda \|\vch\|_{\Hk}^2,
\end{equation}
and consider the  optimization problem
\begin{equation}\label{eqn:opt_J_h}
	\inf_{({a},\vch)\in\Vk}\ \! \Jcal({a},\vch).
\end{equation}
One can easily see that $\Jcal_{\!\Fscr}= \Jcal + \delta_{\Fscr}$, which implies that 
\begin{equation}
\Jcal({a},\vch)\le \Jcal_{\!\Fscr}({a},\vch),
\qquad \forall\!\ ({a},\vch)\in\Vk.	
\end{equation}
Consequently, if \eqref{eqn:opt_J_h} has a solution $({a}^*,\vch^*)$ such that the operator $\Hankel(\vch^*)$ is finite-rank, then $({a}^*,\vch^*)$ is a solution for \eqref{eqn:opt_h_rank_noConstraint} as well. In other words, the identification problem with internal positivity side-information introduced in \eqref{eqn:opt_h_rank} admits a solution.
Hence, we need to study the solution behavior of \eqref{eqn:opt_J_h}.
To this end, we require several technical assumptions.
\begin{assumption}
	\label{ass:dom_pole_exciting}
	There exists $i\in\{0,1,\ldots,\nD-1\}$ such that $\Lu{t_i}(\vcf_{\rho})\ne 0$.
\end{assumption}
Note that if Assumption \ref{ass:dom_pole_exciting} does not hold, then, for all  $i\in\{0,1,\ldots,\nD-1\}$, we have $\Lu{t_i}(\vcf_{\rho})= 0$, which means that the dominant pole is not excited by the input signal $\vcu$. Accordingly, this assumption essentially says that the input signal excites the dominant pole. 
\begin{assumption}
	\label{ass:finite_input}
	There exists $\underline{t}\le 0$ such that $u_t=0$, for $t< \underline{t}$.
\end{assumption}
Assumption \ref{ass:finite_input} is a technical assumption and introduced mainly for the sake of our  mathematical arguments, i.e., to guarantee the continuity of operator $\Lu{t_i}:\Hk\to\Rbb$, for $i\in\{0,1,\ldots,\nD-1\}$. Indeed, this assumption holds in realistic situations such as when the system is initially at rest. 
Based on these assumptions, we can show the existence and uniqueness for the solution of \eqref{eqn:opt_J_h}.
\begin{theorem}\label{thm:J_sol}
	Under Assumptions \ref{ass:dom_pole_exciting} and \ref{ass:finite_input}, optimization problem \eqref{eqn:opt_J_h} admits a unique solution, i.e., there exists $({a}^*,\vch^*)\in\Vk$ such that 
	\begin{equation}\label{eqn:J_h_opt}
		\Jcal({a}^*,\vch^*)< \Jcal({a},\vch),
	\end{equation}
	for any $({a},\vch)\in\Vk\backslash\big\{({a}^*,\vch^*)\big\}$.
\end{theorem}
\begin{proof}
Let set $\Rscr\subset \Vk$ be defined as ${\Rscr} = \bigcap_{s=0}^{\infty}{\Rscr}_s$.	
Accordingly, one can see that $\sum_{s=0}^{\infty}\delta_{{\Rscr}_s} = \delta_{{\Rscr}}$, and hence, we have
\begin{equation}\label{eqn:J_P}
	\Jcal({a},\vch)  =  \Ecal_{\rho}({a},\vch)
	+
	\lambda \|\vch\|_{\Hcal_{\bbk}}^2
	+
	\delta_{{\Rscr}}({a},\vch).
\end{equation}
With respect to each $s\in\Zbb_+$, define set $\Qcal_s\subset\Vk$ as 
\begin{equation}
	\Qcal_s:= 
	\Big\{
	\big({a},(h_s)_{s\in\Zbb_+}\!\big)\in\Vk \!\ \Big| \!\  
	h_s+{a}\rho^s\ge 0
	\Big\}.	
\end{equation}
For any $s\in\Zbb_+$, due to the reproducing property of the kernel, we have
\begin{equation}
\begin{split}
h_s + {a} \rho^s 
&= \inner{\vch}{\kernel_s}_{\Hk}+{a}\rho^s
=\inner{({a},\vch)}{(\rho^s,\kernel_s)}_{\Vk}. 	
\end{split}
\end{equation}
Therefore, we know that $\Qcal_s\subset\Vk$ is a half-space, and hence, it is a non-empty, closed and convex subset of $\Vk$, for all $s\in\Zbb_+$.
Note that $[{a}_{\min},\infty)\times\Hk$ is a non-empty, closed and convex subsets of $\Vk$.
One can see that 
\begin{equation}
	{\Rscr} =  \big(\capOp_{s=0}^{\infty}\Qcal_s\big) \!\ \cap \!\ \big([{a}_{\min},\infty)\times\Hk\big),
\end{equation} 
and also, we know that $({a}_{\min},\zero)$ belongs to $[{a}_{\min},\infty)\times\Hk$ and $\Qcal_s$, for each $s\in\Zbb_+$.
Therefore, ${\Rscr}$ is  a non-empty, closed and convex subset of $\Vk$. Consequently, it follows that
$\delta_{{\Rscr}}:\Vcal\to\Rbb\cup\{+\infty\}$  is a proper, convex and lower semi-continuous function, where we have ${\delta_{{\Rscr}}({a}_{\min},\zero)=0}$.
With respect to each $i\in\{0,1,\ldots, \nD-1\}$, define $\varphi_i$ as 
\begin{equation}\label{eqn:phi_i}
\varphi_i := \sum_{s=0}^{t_i-\underline{t}}u_{t_i-s}\kernel_s
=u_{t_i}\kernel_0 +u_{t_i-1}\kernel_1 +\ldots+u_{\underline{t}}\kernel_{t_i-\underline{t}}.	
\end{equation}
Since 
$\Hk$ is a linear space which contains the sections of the kernel, we know that $\varphi_i\in\Hk$, for each $i\in\{0,1,\ldots,\nD-1\}$.
Moreover, from Assumption \ref{ass:finite_input}, reproducing property of kernel and the linearity property of inner product, it follows that
\begin{equation}\label{eqn:Lu_ti_inner_phi_i}
\begin{split}
\Lu{t_i}(\vch)
&=
\sum_{s=0}^{t_i-\underline{t}}\inner{\vch}{\kernel_s}_{\Hk} u_{t_i-s}
\\&=
\langle\vch,\sum_{s=0}^{t_i-\underline{t}}\kernel_s u_{t_i-s}\rangle_{\Hk}
= \inner{\vch}{\varphi_i}_{\Hk}.
\end{split}
\end{equation}
Accordingly, we have 
\begin{equation}
{a}\Lu{t_i}+\Lu{t_i}(\vch) = 
{a}\Lu{t_i}+\inner{\vch}{\psi_i}_{\Hk} = 
\inner{({a},\vch)}{\psi_i}_{\Vk},
\end{equation}
where $\psi_i\in\Vk$ is the vector defined as $\psi_i = (\Lu{t_i}(\vcf_{\rho}),\varphi_i)$, for $i=0,1,\ldots,\nD-1$.
Therefore, one can see that
\begin{equation*}
	\Jcal({a},\vch)
	\!=\!
	\sum_{i=0}^{\nD-1}\!\big(y_{t_i}\!-\!\inner{({a},\vch)}{\psi_i}_{\Vk}\big)^2 + \lambda\|\vch\|_{\Hk}^2 \!+\! \delta_{{\Rscr}}({a},\vch),
\end{equation*}
Accordingly, since $\delta_{{\Rscr}}$ is proper, convex and lower semi-continuous, and also, due to the fact that
\begin{equation}	
	\Jcal({a}_{\min},\zero)
	= 
	\sum_{i=1}^n
	\big(y_{t_i}-
	{a}_{\min}\Lu{t_i}(\vcf_{\rho})\big)^2 <\infty,	
\end{equation}
we know that $\Jcal$ is a proper, convex and lower semi-continuous function.
This implies that \eqref{eqn:opt_J_h} has a solution \cite{peypouquet2015convex}.
Define bilinear operator $\mxQ:\Vk\times\Vk\to\Rbb$  as 
\begin{equation}
	\begin{split}
		\mxQ\big(&({a}_1,\vch_1),({a}_2,\vch_2)\big)
		\\=&\!\sum_{i=0}^{\nD-1}
		\inner{({a}_1,\vch_1)}{\psi_i}_{\Vk}
		\inner{\psi_i}{({a}_2,\vch_2)}_{\Vk}
		\!+\!
		\lambda\inner{\vch_1}{\vch_2}_{\Hk}.
	\end{split}
\end{equation}
For any $({a},\vch)\in\Vk$, one can easily that 
\begin{equation}
\mxQ\big(({a},\vch),({a},\vch)\big)
= \sum_{i=0}^{\nD-1}
\inner{({a},\vch)}{\psi_i}_{\Vk}^2
\!+\!
\lambda\|\vch\|_{\Hk}^2\ge 0.
\end{equation}
Moreover, since $\lambda$ is a positive real scalar, if $\mxQ\big(({a},\vch),({a},\vch)\big)=0$, then,  we need to have $\vch=0$ and $\inner{({a},\vch)}{\psi_i}_{\Vk}=0$, for all $i=0,1,\ldots,\nD-1$, which implies that ${a} \Lu{t_i}(\vcf_{\rho})=0$.
Subsequently, due to Assumption~\ref{ass:dom_pole_exciting}, we have ${a}=0$, i.e.,  $({a},\vch) = (0,\zero)$. Based on this argument, we know that $\mxQ$ is a positive definite bilinear operator.
Therefore, the function $f:\Vk\to\Rbb$, defined as 
$f(\vcv)=\mxQ(\vcv,\vcv)$, for all $\vcv\in\Vk$, is strictly convex \cite{peypouquet2015convex}.
Note that we have
\begin{equation}
\Jcal({a},\vch) = 
f({a},\vch)-2
\mxL({a},\vch)
+
\sum_{i=0}^{\nD-1}y_{t_i}^2 
+
\delta_{{\Rscr}}({a},\vch),
\end{equation}
where $\mxL:\Vk\to\Rbb$  is the bounded linear operator defined as
\begin{equation}
	\mxL({a},\vch)
	=\sum_{i=0}^{\nD-1}
	y_{t_i}\inner{({a},\vch)}{\psi_i}_{\Vk}.
\end{equation}
Therefore, since $f$ is strictly convex, $\mxL$ is linear, and, $\delta_{{\Rscr}}$ is convex, we know that $\Jcal:\Vk\to\Rbb$ is a strictly convex function, and consequently, the the solution of optimization problem \eqref{eqn:opt_J_h} is unique \cite{peypouquet2015convex}.
This concludes the proof.
\end{proof}
Due to Theorem \ref{thm:J_sol}, 
we know that the convex program \eqref{eqn:opt_J_h} has a unique solution $({a}^*,\vch^*)$.
Meanwhile, one should note that \eqref{eqn:opt_J_h} is an infinite dimensional optimization problem with infinite number of constraints. Thus, obtaining the solution  $({a}^*,\vch^*)$ is not straightforward.
On the other hand, since $\vch^*$ belongs to the set of stable responses $\Hk$ dominated by $\vcf_{\rho}$,  one may intuitively expect that 
$\vch^*\in{\Rscr}_m$, for large enough $m\in\Zbb_+$. 
In other words, the solution to the optimization problem \eqref{eqn:opt_J_h} is determined by a finite number of constraints and the remaining constraints are unnecessary.  
In order to formalize this idea, let function
$\Jcal_m:\Vk\to \RR\cup\{+\infty\}$ be defined as   
\begin{equation}\label{eqn:Jm}
	\Jcal_m({a},\vch) =  \Ecal_{\rho}({a},\vch)
	+
	\sum_{s=0}^{m}\delta_{{\Rscr}_s}({a},\vch)
	+
	\lambda \|\vch\|_{\Hk}^2,
\end{equation}
and, consider the following program
\begin{equation}\label{eqn:min_Jm_h}
		\inf_{({a},\vch)\in\Vk}\ \! \Jcal_m({a},\vch).
\end{equation}
Note that \eqref{eqn:min_Jm_h} is equivalent to
\begin{equation}\label{eqn:min_Jm_constraint}
	\!\!\!\!\!\!\!\!
	\begin{array}{cl}
		\minOp_{{a}\in\Rbb,\vch\in \Hk} \!\!\!&\!\!\! 
		\sum_{i=0}^{\nD-1}
		\!\big(
		y_{t_i}-{a}\Lu{t_i}(\vcf_{\rho})
		-\Lu{t_i}(\vch)
		\big)^2 \!\!+\! \lambda \|\vch\|_{\Hk}^2,
		\!\!\!\!\!\!\!\!\!\!\\
		\text{s.t.}
		\!\!\!&\!\!\!
		h_t+{a}\rho^t\ge 0, \quad  \forall\!\ t\in\{0,1,\ldots, m\},
		\\
		\!\!\!&\!\!\!
		{a}\ge{a}_{\min}.
	\end{array} 
\end{equation}
The next theorem guarantees the existence and uniqueness of solution for optimization problem \eqref{eqn:min_Jm_h}, or equivalently, for program \eqref{eqn:min_Jm_constraint}.
\begin{theorem}\label{thm:Jm_sol}
	Under the assumptions of Theorem \ref{thm:J_sol}, 
	for each $m\in\Zbb_+$, problem \eqref{eqn:min_Jm_h} admits unique solution  $({a}^{\mt},\vch^{\mt})$. 
\end{theorem}
\begin{proof}
Define set $\Pscr^{\mt}\subset\Vk$ as $\Pscr^{\mt} = \bigcap_{s=0}^{m}{\Rscr}_s$.
By replacing $\Pscr$ with $\Pscr^{\mt}$ in the proof of Theorem \ref{thm:J_sol} and then repeating the same steps, the claim follows. 
\end{proof}
Once the existence and uniqueness for the solution of \eqref{eqn:min_Jm_h} are established by Theorem \ref{thm:Jm_sol}, a reasonable concern is
the asymptotic behavior of $({a}^{\mt},\vch^{\mt})$, especially with respect to $({a}^*,\vch^*)$.
The next theorem reveals this link saying that the solution $({a}^{\mt},\vch^{\mt})$ coincides with $({a}^{*},\vch^*)$ when $m$ is large enough.

\begin{theorem}\label{thm:Jm_star_sol}
Let Assumptions 
\ref{ass:dom_pole_exciting} and \ref{ass:finite_input} hold. Then, the following statements hold:\\	
\emph{i)} 
Under Assumption \ref{ass:kernel_dominated}, there is a non-negative integer $m$ such that ${a}^{\mt} = {a}^{*}$ and $\vch^{\mt} = \vch^{*}$.\\
\emph{ii)} For a non-negative integer $m$, one has $({a}^{\mt},\vch^{\mt}) = ({a}^{*},\vch^*)$ if and only if  $\Jcal({a}^{\mt},\vch^{\mt})<\infty$.\\
\emph{iii)} 
If $({a}^{\mt},\vch^{\mt}) = ({a}^{*},\vch^*)$, for a non-negative integer $m$, then ${a}^{\mbt} = {a}^{*}$ and $\vch^{\mbt} = \vch^{*}$, for every $\ol{m}\ge m$.
\end{theorem}
\begin{proof}
Part i)	
For any ${a}\ge{a}_{\min}$ and any $m\in\Zbb_+$, we have
\begin{equation}
\Jcal_m({a},\zero) = 
\Jcal({a},\zero) = 
\sum_{i=1}^n
\big(y_{t_i}-
{a}\Lu{t_i}(\vcf_{\rho})\big)^2 \in [0,\infty).	
\end{equation}
Let ${a}_0$ be
$\argmin_{{a}\ge{a}_{\min}}\Jcal({a},\zero)$,
and, define $C_0$ as $C_0 := \Jcal({a}_0,\zero)$, which is equal to $\Jcal_m({a}_0,\zero)$, for any $m\in\Zbb_+$.
If $C_0=0$, then one can easily see that $({a}_0,\zero)$ is a solution for optimization problems \eqref{eqn:opt_J_h} and \eqref{eqn:min_Jm_h}, where, because of their uniqueness, $({a}_0,\zero)=({a}^{\mt},\vch^{\mt})=({a}^{*},\vch^{*})$.
Now, we consider the case $C_0>0$. 
Define $m_0$ as following
\begin{equation}\label{eqn:m_0}
	m_0 =\min\bigg\{
	m\in\Zbb_+\!\ \bigg|\!\ 
	m\ge \frac{1}{2}
	\frac{\ln(C_0C)-\ln({a}_{\min}^2\lambda)}{\ln(\rho)-\ln(\rhozero)}\bigg\}.
\end{equation}
Based on the definition of $m_0$, one can easily see that
\begin{equation}\label{eqn:pf_m_finite_03}
{a}_{\min}\rho^s
-\lambda^{-\frac{1}{2}}C_0^{\frac{1}{2}}C^{\frac{1}{2}}\rhozero^s \ge 0,
\end{equation}
for any $s\ge m_0$. 
Let $m$ be an arbitrary integer such that $m\ge m_0$, 
and, consider the convex program \eqref{eqn:min_Jm_h} with unique solution $({a}^{\mt},\vch^{\mt})$. 
We know that $h_s^{\mt}+{a}^{\mt}\rho^s\ge 0$, for each $s=0,1,\ldots,m$. 
On the other hand, for $s > m$, from the reproducing property and the Cauchy-Schwartz inequality,  it follows that
\begin{equation}\label{eqn:pf_m_finite_01}
\begin{split}
h_s^{\mt}+{a}^{\mt}\rho^s 
&\ge \inner{\vch^{\mt}}{\kernel_s}_{\Hk} + {a}^{\mt}\rho^s
\\&
\ge - \|\vch^{\mt}\|_{\Hk}\kernel(s,s)^{\frac{1}{2}}+ {a}_{\min}\rho^s.
\end{split}
\end{equation}
Note that due to $\Jcal_m({a}^{\mt},\vch^{\mt})\le \Jcal_m({a}_0,\zero)$, one has
\begin{equation}
	\lambda \|\vch^{\mt}\|_{\Hk}^2\le \Jcal_m({a}^{\mt},\vch^{\mt})\le C_0,
\end{equation}
which implies that $\|\vch^{\mt}\|_{\Hk}\le \lambda^{-\frac{1}{2}}C_0^{\frac{1}{2}}$.
Hence, according to \eqref{eqn:pf_m_finite_03}, \eqref{eqn:pf_m_finite_01} and Assumption \ref{ass:kernel_dominated},
we have
\begin{equation}\label{eqn:pf_m_finite_02}
h_s^{\mt}+{a}^{\mt}\rho^s 
\ge 
{a}_{\min}\rho^s
-\lambda^{-\frac{1}{2}}C_0^{\frac{1}{2}}C^{\frac{1}{2}}\rhozero^s\ge 0,
\end{equation}  
which implies that $h_s^{\mt}+{a}^{\mt}\rho^s\ge 0$, for all $s\in\Zbb_+$.
Therefore, we have
$\sum_{s \ge 1}\delta_{{\Rscr}_s}({a}^{\mt},\vch^{\mt})=0$, and subsequently, it follows that 
$\Jcal({a}^{\mt},\vch^{\mt}) = \Jcal_m({a}^{\mt},\vch^{\mt})$.
On the other hand, according to the definition of $({a}^*,\vch^*)$ and $({a}^{\mt},\vch^{\mt})$, we know that 
$\Jcal_m({a}^{\mt},\vch^{\mt})\le \Jcal_m({a}^{*},\vch^{*})$
and
$\Jcal({a}^{*},\vch^{*})\le \Jcal({a}^{\mt},\vch^{\mt})$.
Since $\Jcal_m({a},\vch)\le \Jcal({a},\vch)$, for all $({a},\vch)\in\Vk$, one can see that
\begin{equation}
		\Jcal_m({a}^{\mt},\vch^{\mt})\le 
		\Jcal_m({a}^{*},\vch^{*})\le
		\Jcal_m({a}^{\mt},\vch^{\mt}).
\end{equation}
Hence, we have $\Jcal_m({a}^{\mt},\vch^{\mt})=\Jcal_m({a}^{*},\vch^{*})$, and, subsequently, due to Theorem \ref{thm:Jm_sol}, one has $({a}^{\mt},\vch^{\mt})=({a}^{*},\vch^{*})$.
This concludes the proof of part i). 
\\
Part ii)
Consider the case that $\Jcal({a}^{\mt},\vch^{\mt})<\infty$. 
This implies that $\sum_{s \ge 1}\delta_{{\Rscr}_s}({a}^{\mt},\vch^{\mt})=0$, and consequently, $\Jcal({a}^{\mt},\vch^{\mt}) = \Jcal_m({a}^{\mt},\vch^{\mt})$. 
One can easily see that, 
\begin{equation}
	\Jcal_m({a},\vch)\le \Jcal_{\ol{m}}({a},\vch)\le \Jcal({a},\vch),
\end{equation}
for any $m\le \ol{m}$ and each $({a},\vch)\in\Vk$. 
Accordingly, due to the definition of $({a}^*,\vch^*)$ and $({a}^{\mt},\vch^{\mt})$, we have
\begin{equation}
\begin{split}
	\Jcal_m({a}^{\mt},\vch^{\mt})&\le 
	\Jcal_m({a}^{*},\vch^{*})\le
	\Jcal({a}^{*},\vch^{*})
	\\&\ \ \qquad\le
	\Jcal({a}^{\mt},\vch^{\mt})=
	\Jcal_m({a}^{\mt},\vch^{\mt}),
\end{split}
\end{equation}
which implies that $({a}^{*},\vch^{*})$ is a solution for \eqref{eqn:min_Jm_h}. 
Since, this solution is unique according to Theorem \ref{thm:Jm_sol}, we need to have $({a}^{\mt},\vch^{\mt}) = ({a}^{*},\vch^*)$. The converse is straightforward, and  concludes the proof of part ii). 
\\Part iii) 
From the previous part, we know that $\Jcal({a}^{\mt},\vch^{\mt})<\infty$. 
Consequently, we have $\sum_{s \ge 1}\delta_{{\Rscr}_s}({a}^{\mt},\vch^{\mt})=0$ which implies that
\begin{equation}
\Jcal_m({a}^{\mt},\vch^{\mt}) = 
\Jcal_{\ol{m}}({a}^{\mt},\vch^{\mt}) = 
\Jcal({a}^{\mt},\vch^{\mt}).	
\end{equation}
Accordingly, due to the definition of $({a}^{\mt},\vch^{\mt})$ and $({a}^{\mbt},\vch^{\mbt})$, we have
\begin{equation}
	\begin{split}
	\Jcal_{\ol{m}}({a}^{\mt},\vch^{\mt})
	&= 
	\Jcal_m({a}^{\mt},\vch^{\mt})
	\\&\le 
	\Jcal_m({a}^{\mbt},\vch^{\mbt})\le
	\Jcal_{\ol{m}}({a}^{\mbt},\vch^{\mbt}).
	\end{split}
\end{equation}
Therefore, $({a}^{\mt},\vch^{\mt})$ is the unique solution of optimization problem 
$\inf_{({a},\vch)\in\Vk}\ \! \Jcal_{\ol{m}}({a},\vch)$, and consequently, we have
\begin{equation}
({a}^{\mbt},\vch^{\mbt}) = ({a}^{\mt},\vch^{\mt}) = ({a}^{*},\vch^{*}),	
\end{equation}
where the second equality holds by assumption.
This concludes the proof of part iii) and the proof of Theorem \ref{thm:Jm_star_sol}.
\end{proof} 
The following observation is a direct result of Theorem \ref{thm:Jm_star_sol}.
\begin{corollary}\label{cor:m_star}
	Under the assumptions of Theorem \ref{thm:Jm_star_sol}, there exists a non-negative integer  $m^*$ such that $({a}^{\mt},\vch^{\mt}) = ({a}^{*},\vch^*)$ if and only if $m\ge m^*$.
	Indeed, for any $m<m^*$, there exists $s\in\Zbb_+$ such that 
	$h_s^{\mt}+{a}^{\mt}\rho^s <0$.
	This implies that
	\begin{equation}
		m^* =
		\min\big\{
		m\in\Zbb_+\!\ \big|\!\ 
		\Jcal({a}^\mt,\vch^\mt)<\infty\big\}.
	\end{equation}
	Moreover, for $m_0$ introduced in \eqref{eqn:m_0}, we have $m^*\le m_0$.
\end{corollary}
\begin{corollary}\label{rem:ul_s}
Let $m\in\Zbb_+$ be such that the impulse response $\vcg^\mt := (h_s^\mt+{a}^\mt\rho^s)_{s=0}^\infty$ is not non-negative, i.e., $m<m^*$.
Then, based on the proof of Theorem \ref{thm:Jm_star_sol}, one can see that there exists $\underline{s}\in\Zbb_+$ such that $g_{\underline{s}}<0$ and $s<m_0$, where $m_0$ is introduced in \eqref{eqn:m_0}.
\end{corollary}
Due to Theorem \ref{thm:Jm_star_sol} and Corollary \ref{cor:m_star}, it suffices to consider only a finite number of constraints in optimization problem \eqref{eqn:opt_J_h} as in
\eqref{eqn:min_Jm_h}.
The feasible set of this optimization problem is infinite-dimension which makes the problem intractable at the current format.
In the remainder of this section, we derive a practical heuristic for obtaining the solution of  \eqref{eqn:opt_h_rank}. To this end, we need to introduce the \emph{representer theorem} and an additional definition.

\begin{theorem}[Representer Theorem, \cite{scholkopf2001generalized, dinuzzo2012representer}]\label{thm:rep_thm}
	 Let $e:\Rbb^n\to \Rbb\cup\{+\infty\}$ be a given function, $\Hcal$ be a Hilbert space endowed with inner product $\inner{\cdot}{\cdot}_{\Hcal}$, and $\vcw_1,\ldots,\vcw_n$ be given vectors in $\Hcal$. Consider the following optimization problem 
	\begin{equation}
		\label{eqn:opt_representer_thm}
		\min_{\vcw\in\Hcal} \    
		e(\inner{\vcw}{\vcw_1}_{\Hcal}, \ldots,
		\inner{\vcw}{\vcw_n}_{\Hcal})+ \kappa(\|\vcw\|_{\Hcal}),
	\end{equation}
	where $\kappa:\Rbb_+\to \Rbb$ is an increasing function.
	If \eqref{eqn:opt_representer_thm} has a solution,
	there also exists a solution in the linear span of $\vcw_1,\ldots,\vcw_n$.	
\end{theorem}
Before proceeding to the next theorem, we define
the matrices $\mxO\in \Rbb^{\nD\times \nD}$, 
$\mxL\in \Rbb^{\nD\times (m+1)}$ and 
$\mxK\in \Rbb^{(m+1)\times (m+1)}$ as following
\begin{equation}\label{eqn:OLK}
\!\!\!\!\!\!
	\begin{array}{ll}
		\mxO(i,j)  = 
		\Lu{t_{i-1}}(\Lu{t_{j-1}}(\kernel)), & 1\!\le\! i,j\!\le\! \nD, \\
		\mxL(i,j)  = 
		\Lu{t_{i-1}}(\kernel_{j-1}),
		& 1\!\le\! i\!\le\! \nD,1\!\le\! j\!\le\! m \! +\! 1, \\
		\mxK(i,j) = \kernel(i-1,j-1),  
		& 1 \!\le\! i,j \!\le\! m \! +\! 1. \\
	\end{array}
\!\!\!\!\!\!
\end{equation}
Also, the vectors $\vcy\in\Rbb^{\nD}$, $\vcb\in\Rbb^{\nD}$ and $\vcc\in\Rbb^{m+1}$ are defined respectively as
\begin{equation}\label{eqn:vec_y_b}
	\vcy := \begin{bmatrix}y_{t_i}\end{bmatrix}_{i=0}^{\nD-1}\,,
	\quad
	\vcb := \begin{bmatrix}\Lu{t_i}(\vcf_{\rho})\end{bmatrix}_{i=0}^{\nD-1}\,,	
\end{equation}
and
\begin{equation} \label{eqn:vec_c}
	\vcc := \begin{bmatrix}\rho^j\end{bmatrix}_{j=0}^{m}\,. 
\end{equation}

\begin{theorem}\label{thm:rep_posi_id}
Let Assumption \ref{ass:finite_input} hold.
Then, for any non-negative integer $m$, there exists 
$\vcx^\mt = [x_0^\mt,\ldots,x_{\nD+m}^\mt]^\tr\in\Rbb^{\nD+m+1}$ such that  the unique solution of \eqref{eqn:min_Jm_constraint}, $\vch^\mt=(h_t^{\mt})_{t=0}^{\infty}$, admits the following parametric representation 
\begin{equation}\label{eqn:hm_rep}
h^{\mt}_t = 
\sum_{i=0}^{\nD} x_i^\mt \Lu{t_i}(\kernel_t) + 
\sum_{s=0}^m  x_{\nD+s}^\mt \kernel_s(t). 
\end{equation}	
Moreover, $({a}^\mt,\vcx^\mt)$ is the solution of the following convex quadratic program
\begin{equation}\label{eqn:min_rho_x_finite_dim}
\!\!\!\!\!\!\!
\begin{array}{cl}
	\!\!\!
	\minOp_{{a}\in\Rbb,\!\ \vcx\in \Rbb^{\nD+m+1}}
	\!\!\!\!&\!\!
	\big\|
	\vcy
	-\vcb{a}
	-\begin{bmatrix}\mxO \!\!&\!\!  \mxL\end{bmatrix}
	\vcx
	\big\|^2 
	\!+\!
	\lambda
	\vcx^\tr
	\!\!
	\begin{bmatrix}
		\mxO 		\!\!&\!\!	\mxL\\
		\mxL^\tr	\!\!&\!\!	\mxK\\
	\end{bmatrix}
	\vcx,
	\!\!
	\\
	\!\!\!
	\mathrm{s.t.}
	\!\!\!\!&\!\! 
	\begin{bmatrix}
		\mxL^\tr \!\!&\!\! \mxK\\
	\end{bmatrix}
	\vcx+ \vcc{a}\ge 0,\\
	\!\!\!&\!\!
	{a}\ge {a}_{\min}.
\end{array} 
\!\!\!\!\!\!\!\!
\end{equation}	
\end{theorem}
\begin{proof}
For $s=0,\ldots,m$, let $\Acal_s$ be the following set
\begin{equation}
	\Acal_s=\big\{({a},x)\in\Rbb^2\!\ \big| \!\ x+{a}\rho^s\ge 0,\ {a}\ge {a}_{\min}\big\},	
\end{equation}
and, define the function $e:\Rbb^{\nD+m+1}\to \Rbb\cup\{+\infty\}$ such that for any $\vcx\in\Rbb^{\nD+m+1}$ we have
\begin{equation}\label{eqn:e}	
\begin{split} \!\!\!\!\! \!\!\!\!\! 	
\!
e&(x_0,\ldots,x_{\nD+m})=
\\
\!
&
\min_{{a}\in\Rbb}
\Bigg[
\!
\sum_{i=0}^{\nD-1}(y_{t_i}-{a}\Lu{t_i}(\vcf_{\rho})-x_i)
\!+\!
\sum_{s=0}^m\delta_{\Acal_s}\!({a},x_{\nD+s})
\Bigg].	\!\!\!\!\! 
\end{split}
\end{equation}
Also, for $i=0,\ldots,\nD-1$ and $s=0,\ldots,m$,
let $\varphi_i$ and $\varphi_{\nD+s}$ be defined respectively as in \eqref{eqn:phi_i} and $\varphi_{\nD+s}=\kernel_s$. 
Due to the reproducing property and \eqref{eqn:Lu_ti_inner_phi_i}, we know that
$\Lu{t_i}(\vch) = \inner{\vch}{\varphi_i}_{\Hk}$ and $h_s = \inner{\vch}{\kernel_s}_{\Hk}$,
for any $\vch = (h_s)_{s=0}^\infty$.
Accordingly, due to \eqref{eqn:e}, one can easily see that \eqref{eqn:min_Jm_h} is equivalent to the following optimization problem
\begin{equation}
	\min_{\vch\in\Hk}e(\inner{\vch}{\varphi_0}_{\Hk},\ldots,\inner{\vch}{\varphi_{\nD+m}}_{\Hk})
	+
	\lambda\|\vch\|_{\Hk}^2,
\end{equation}
with unique solution $\vch^\mt$.
Therefore, due to Theorem \ref{thm:rep_thm}, we know that  $\vch^\mt$ belongs to the span of $\varphi_0,\varphi_1,\ldots,\varphi_{\nD+m}$, i.e., there exists  $\vcx^\mt = [x_0^\mt,\ldots,x_{\nD+m}^\mt]^\tr\in\Rbb^{\nD+m+1}$ such that 
\begin{equation}\label{eqn:hm_sum_phi_i}
\vch^{\mt} = 
\sum_{i=0}^{\nD+m}x_i^\mt\varphi_i =
\sum_{i=0}^{\nD-1}x_i^\mt\varphi_i +
\sum_{s=0}^{m}x_{\nD+s}^\mt\kernel_s.
\end{equation}
Due to reproducing property, we know that $h_t^{\mt} = \inner{\vch^\mt}{\kernel_t}_{\Hk}$, for any $t\in\Zbb_+$.
Accordingly, from \eqref{eqn:Lu_ti_inner_phi_i}, \eqref{eqn:hm_sum_phi_i}, the linearity property of inner product, and the reproducing property, we have
\begin{equation}\label{eqn:Lu_hm_t}
\begin{split}
	h^{\mt}_t &=
	\big\langle\sum_{i=0}^{\nD-1}x_i^\mt\varphi_i +
	\sum_{s=0}^{m}x_{\nD+s}^\mt\kernel_s,\!\, \kernel_t \big\rangle_{\Hk} 
	\\&=	
	\sum_{i=0}^{\nD-1}x_i^\mt\inner{\varphi_i}{\kernel_t}_{\Hk} +
	\sum_{s=0}^{m}x_{\nD+s}^\mt\inner{\kernel_s}{\kernel_t}_{\Hk} 
	\\&=\sum_{i=0}^{\nD} x_i^\mt \Lu{t_i}(\kernel_t) + 
	\sum_{s=0}^m  x_{\nD+s}^\mt \kernel_s(t). 
\end{split}	
\end{equation}
Moreover, for $j=0,\ldots,\nD-1$, we have
\begin{equation}\label{eqn:Lu_tj_hm}
\begin{split}
	\Lu{t_j}(\vch^\mt)
	&= \Lu{t_j}(\sum_{i=0}^{\nD-1}x_i^\mt\varphi_i +
	\sum_{s=0}^{m}x_{\nD+s}^\mt\kernel_s)
	\\
	&=\sum_{i=0}^{\nD-1}x_i^\mt\Lu{t_j}(\varphi_i) +
	\sum_{s=0}^{m}x_{\nD+s}^\mt\Lu{t_j}(\kernel_s)
	\\
	&=\sum_{i=0}^{\nD-1}x_i^\mt\Lu{t_j}(\Lu{t_i}(\kernel)) +
	\sum_{s=0}^{m}x_{\nD+s}^\mt\Lu{t_j}(\kernel_s).
\end{split}	
\end{equation} 
Considering optimization problem \eqref{eqn:min_Jm_constraint}, which is equivalent to \eqref{eqn:min_Jm_h}, we replace $\vch$ with the parametric form given in \eqref{eqn:hm_sum_phi_i}.
Then, due to \eqref{eqn:OLK}, \eqref{eqn:Lu_hm_t}, \eqref{eqn:Lu_tj_hm}, and the definition of vectors $\vcb$, $\vcc$ and $\vch^\mt$,  the optimization problem \eqref{eqn:min_rho_x_finite_dim} follows.
This concludes the proof.
\end{proof}
\begin{remark}\label{rem:OLK_at_rest}\normalfont
Let the system be initially at rest and the sampling times be $\Tscr=\{0,1,\ldots,\nD-1\}$.
With respect to each $n_1,n_2\in\Zbb_+$, we define matrix $\mxK_{n_1,n_2}\in\Rbb^{n_1\times n_2}$ such that $\mxK_{n_1,n_2}(i,j)={\kernel(i-1,j-1)}$, for $i=1,\ldots,n_1$ and $j=1,\ldots,n_2$.
Then, one can easily see that 
\begin{equation}
	\mxO = \mxT_{\vcu}\mxK_{\nD,\nD}\mxT_{\vcu}^\tr\,,
	\quad
	\mxL = \mxT_{\vcu}\mxK_{\nD,m+1}\,, 
\end{equation}
and
\begin{equation}
	\mxK = \mxK_{m+1,m+1}\,,
\end{equation}
where $\mxT_{\vcu}\in\Rbb^{n\times n}$ is the Toeplitz matrix defined as $\mxT_{\vcu}=[u_{i-j}]_{i,j=1}^{\nD}$.
\end{remark}
Theorem \ref{thm:rep_posi_id} offers a practical way to solve problem  \eqref{eqn:Jm}.
Due to Theorem \ref{thm:Jm_star_sol} and Corollary \ref{cor:m_star}, we know that this solution coincides with the solution of \eqref{eqn:J} provided that $m\ge m^*$.
Nevertheless, compared to \eqref{eqn:J}, the main optimization problem  \eqref{eqn:J_F} has an additional constraint on  the rank of resulting Hankel operator being finite.
In the remainder of this section, we fill this gap by employing the notion of finite Hankel rank kernels.

\begin{definition}\label{def:finite_Hankel_rank_kernel}
	We call the Mercer kernel $\kernel:\Zbb_+\times\Zbb_+$ a \emph{finite Hankel rank} if, for each $s\in\Zbb_+$, the section of kernel at $s$ is a finite Hankel rank impulse response, i.e., 
	\begin{equation}
		\rank(\Hankel(\kernel_s)) < \infty, \quad \forall\!\ s\in\Zbb_+.
	\end{equation}
\end{definition}
The commonly employed stable kernels in the literature \cite{pillonetto2014kernel, dinuzzo2015kernels} are 
\emph{Tuned/Correlated} (TC), 
\emph{Diagonal/Correlated} (TC), and 
\emph{Stable Spline} (TC),
which are respectively denoted by $\kernelTC$, $\kernelDC$ and $\kernelSS$, 
and defined as  follows:
\begin{itemize}
	\item Tuned/Correlated (TC) kernel:
	\begin{equation}\label{eqn:kernelTC}
		\kernelTC(s,t) = 
		\beta^{\max(s,t)}, \quad \forall s,t\in \Zbb_+,
	\end{equation}
	where $\beta\in [0,1)$ is the hyperparameter of $\kernelTC$,
	\item Diagonal/Correlated (DC) kernel:
	\begin{equation}\label{eqn:kernelDC}
		\kernelDC(s,t) = 
		\beta^{\frac{s+t}{2}}\gamma^{|s-t|}, \quad \forall s,t\in \Zbb_+,
	\end{equation}
	where $\beta \in [0,1)$ and $\gamma\in [-1,1]$ are the hyperparameters of $\kernelDC$, and, 
	\item Stable Spline (SS) kernel:
	\begin{equation}
		\kernelSS(s,t) = 
		\frac{\beta^{s+t+\max(s,t)}}{2}
		-
		\frac{\beta^{3\max(s,t)}}{6},
		\quad \forall s,t\in \Zbb_+,
	\end{equation}
	where $\rho\in [0,1)$ is the hyperparameter of $\kernelSS$.
\end{itemize}
As shown by the next theorem, these common kernels are finite Hankel rank.
\begin{theorem}\label{thm:kernel_TC_DC_SS_are_FHR}
	Every finite support kernel, 
	$\kernelDC$, $\kernelTC$ and $\kernelSS$  is finite Hankel rank.
\end{theorem}
\begin{proof}
See Appendix \ref{sec:pf_kernel_TC_DC_SS_are_FHR}.
\end{proof}

Based on the notion of finite Hankel rank kernel and our previous discussion, we can show when the solution of our estimation problem \eqref{eqn:opt_h_rank} can be obtained by solving \eqref{eqn:min_rho_x_finite_dim}.

\begin{theorem}\label{thm:rank_Hankel}
Under the assumptions of Theorem \ref{thm:Jm_star_sol}, if kernel $\kernel$ is finite Hankel rank, then the unique solution of \eqref{eqn:min_Jm_h}  
satisfies
\begin{equation}\label{eqn:rank_hm_finite}
\rank\big(\Hankel(\vch^\mt)\big) < \infty.	
\end{equation}
Moreover, $({a}^\mt,\vch^\mt)$ is a solution of \eqref{eqn:opt_h_rank} provided that $m\ge m^*$. 
\end{theorem}
\begin{proof}
Due to \eqref{eqn:phi_i} and \eqref{eqn:hm_sum_phi_i},
we know that
\begin{equation}\label{eqn:hm_sum_kernel_setion}
\begin{split}
\vch^{\mt} 
=
\sum_{i=0}^{\nD-1}x_i^\mt\sum_{s=0}^{t_i-\underline{t}}u_{t_i-s}\kernel_s
 +
\sum_{s=0}^{m}x_{\nD+s}^\mt\kernel_s.
\end{split}	
\end{equation}
Rearranging the terms in \eqref{eqn:hm_sum_kernel_setion}, 	
one can see that there exists real scalars $\ol{x}_0,\ldots,\ol{x}_{\ol{t}}$,
where $\ol{t} = \max\{m,t_{\nD-1}-\underline{t}\}$,
such that we have
\begin{equation}\label{eqn:h_m_sum_k}
\vch^{\mt} = \sum_{s=0}^{\ol{t}}\ol{x}_s\kernel_s
=
\ol{x}_0\kernel_0+\ldots+\ol{x}_{\ol{t}}\kernel_{\ol{t}}.	
\end{equation}
Therefore, we have
\begin{equation}
	\Hankel(\vch^{\mt})= \sum_{s=0}^{\ol{t}}\ol{x}_s
	\Hankel(\kernel_s),
\end{equation}
and subsequently, it follows that
\begin{equation}\label{eqn}
	\rank(\Hankel(\vch^{\mt}))\le \sum_{s=0}^{\ol{t}}
	\rank(\Hankel(\kernel_s))<\infty,
\end{equation}
i.e., \eqref{eqn:rank_hm_finite} holds.
For $m\ge m^*$, we know that $({a}^\mt,\vch^\mt)$ is the solution of \eqref{eqn:opt_J_h}. Accordingly, due to \eqref{eqn:rank_hm_finite}, we have
\begin{equation}\label{eqn:J_F_hm_J_hm}
\Jcal_{\!\Fscr}({a}^\mt,\vch^\mt) = \Jcal({a}^\mt,\vch^\mt)<\infty.
\end{equation}
On the other hand, for any $({a},\vch)\in\Vk$, one can see that 
\begin{equation}
\Jcal({a}^\mt,\vch^\mt) \le \Jcal({a},\vch) \le \Jcal_{\!\Fscr}({a},\vch).
\end{equation}
Hence, due to \eqref{eqn:J_F_hm_J_hm}, we have
$\Jcal_{\!\Fscr}({a}^\mt,\vch^\mt)\le \Jcal_{\!\Fscr}({a},\vch)$, for any  $({a},\vch)\in\Vk$,
i.e., $({a}^\mt,\vch^\mt)$ is a solution of \eqref{eqn:opt_h_rank}.
This concludes the proof.
\end{proof}
Based on the above discussion,
in order to solve the optimization problem 
\eqref{eqn:opt_h_rank}, it suffices to find the solution of the quadratic program \eqref{eqn:min_rho_x_finite_dim}
where $m \ge m^*$ and $\kernel$ is a given finite Hankel rank kernel such as $\kernelTC$, $\kernelDC$, or $\kernelSS$.
Corollary \ref{cor:m_star} provides a bound for $m^*$. In some special cases, we can provide a more practical bound.
\begin{theorem}\label{thm:mstar_TC}
Let the assumptions of Theorem \ref{thm:Jm_star_sol} hold. If $\kernel$ is either the TC kernel \eqref{eqn:kernelTC} or the DC kernel \eqref{eqn:kernelDC}, then we have $m^*\le t_{\nD-1}-\underline{t}+1$.
\end{theorem}
\begin{proof}
Let $m = t_{\nD-1}-\underline{t}+1$, and consider $({a}^\mt,\vch^\mt)$, the unique solution of \eqref{eqn:min_Jm_h}. For $t=0,1,\ldots,m$, we know that $h_t^\mt+{a}^\mt\rho^t\ge 0$. One the other hand, due to the definition of the TC kernel and \eqref{eqn:hm_sum_kernel_setion}, we have $h_t^\mt = \beta^{t-m}h_m^\mt$, for any $t>m$. 
Note that, due to Assumption \ref{ass:kernel_dominated},  we have $\beta^s\le C \rhozero^s$, for all $s\in\Zbb_+$. This implies that $\beta\le \rhozero^2$.
Therefore, since $\rhozero^2<\rho$, we have
\begin{equation}
\begin{split}
h_t^\mt+{a}^\mt\rho^t 
&= 
\beta^{t-m}h_m^\mt + {a}^\mt\rho^{m}\rho^{t-m}
\\&\ge 
\beta^{t-m}h_m^\mt + \beta^{t-m}{a}^\mt\rho^{m}\ge 0,
\end{split}
\end{equation}
where the last equality is due to $h_m^\mt+{a}^\mt\rho^m\ge 0$.
Therefore, $({a}^\mt,\vch^\mt)$ is feasible for \eqref{eqn:opt_J_h}, and subsequently, we have $\Jcal({a}^\mt,\vch^\mt)<\infty$.
Hence, according to Theorem \ref{thm:Jm_star_sol}, we know that $({a}^\mt,\vch^\mt)=({a}^\mt,\vch^\mt)$ which implies that $m^*\le m =  t_{\nD-1}-\underline{t}+1$.
Based on a similar argument, one can show the same result for the case of DC kernel.
This concludes the proof.
\end{proof}
\begin{remark}
For the  special cases $\kernelTC$ and $\kernelDC$, one can relax Assumption~\ref{ass:kernel_dominated} to $|\kernel(t,t)|\le\rhozero^t$, for all $t\in\Zbb_+$.
\end{remark}
Based on the above discussion, we present a practical implementation scheme for identifying a linear internally positive dynamical system. The algorithm and further details are provided in the next section.

\section{Numerical Implementation Algorithm} \label{sec:alg}
Based on the discussion in Sections \ref{sec:positive_realization} and \ref{sec:Towards_a_Tractable_Solution}, to identify system $\Scal$ with internal positivity side-information, 
we need to solve convex quadratic program \eqref{eqn:min_rho_x_finite_dim}.
This optimization problem can be solved using standard off-the-shelf solvers such as \MATLAB's \texttt{quadprog} or \texttt{cvx} supported by \textsc{MOSEK} \cite{CVX}.
Note that \eqref{eqn:min_rho_x_finite_dim} depends on non-negative integer parameter $m$ which is supposed be larger or equal to the parameter $m^*$ introduced in Corollary~\ref{cor:m_star}.
One possible approach is to set the value of $m$ to $m_0$, which is introduced in \eqref{eqn:m_0} and guaranteed to have property.
Also, one may take initially $m$ equal to $\nD-\underline{t}$ and iteratively increase it until $m$ exceeds $m^*$ and a non-negative impulse response is obtained.  
This is of special interest when a suitable quadratic program solver with a warm-starting feature is available \cite{gondzio2016large}.
In this iterative approach, at each iteration $m$, one should check whether the estimated impulse response $\vcg^\mt=(g^\mt_s)_{s=0}^{\infty}$ is non-negative, or equivalently, $m \ge m^*$. According to Corollary \ref{rem:ul_s}, for verifying this stopping condition, it is enough to see whether $g_s^\mt\ge 0$ holds, for $s<m_0$. Due to \eqref{eqn:m_0}, we know that $m_0$ depends logarithmically on the parameters of problem, and thus, the size of $m_0$ is not prohibitive in the practical examples.
According to Theorem \ref{thm:mstar_TC}, when the TC kernel is employed and the initial value of $m$ is set to $\nD-\underline{t}$, the introduced iterative scheme takes only a single iteration. 
The outline of this approach is summarized in Algorithm \ref{alg:PosiID}.

\begin{algorithm}[t]
	\caption{System Identification with Internal Positivity Side-Information}\label{alg:PosiID}
	\begin{algorithmic}[1]
		\State \textbf{Input:} Set of data $\Dscr$, finite Hankel rank kernel stable $\kernel$, dominant pole $\rho$,  regularization weight $\lambda$, and, $\Delta_m \in\Nbb$.
		\State $m\  \reflectbox{$\to$}\ t_{\nD-1}-\underline{t}+1$ 
		\While{stopping/exiting condition is not met}
		\State Calculate vectors $\vcy$, $\vcb$ and $\vcc$ based on \eqref{eqn:vec_y_b} and \eqref{eqn:vec_c}.
		\State Obtain matrices $\mxO$, $\mxL$ and $\mxK$ as in  \eqref{eqn:OLK}, 	or by Remark~\ref{rem:OLK_at_rest}
		\State Solve quadratic program \eqref{eqn:min_rho_x_finite_dim} for ${a}^\mt$ and $\vcx^\mt$.
		\State Obtain $\vch^\mt$ based on \eqref{eqn:hm_rep}, or equivalently, \eqref{eqn:hm_sum_kernel_setion}.
		\State $\vcg^\mt \  \reflectbox{$\to$}\  ({a}^\mt\rho^s+h_s^\mt)_{s=0}^{\infty}$ 
		\If{$\vcg^\mt$ is non-negative,} exit the loop,
		\Else{$\quad m \  \reflectbox{$\to$}\ m + \Delta_m$}
		\EndIf
		\State \textbf{end}
		\EndWhile
		\State \textbf{end} 
		\State \textbf{Input:} Internally positive impulse response $\vcg^*$, and also $\vcx^*$, ${a}^*$, $\vch^*$.
	\end{algorithmic}
\end{algorithm}

In order to initialize Algorithm \ref{alg:PosiID}, in addition to the set of data $\Dscr$, we need a suitable kernel $\kernel$ and also an estimation of the dominant pole $\rho$ and the regularization weight $\lambda$. 
In general, deciding on the type of kernel depends on the shape and smoothness of the impulse response to be identified.  
Once the type of kernel $\kernel$ is set, we need to estimate the vector of hyperparameters $\theta_{\kernel}$ which characterizes $\kernel$.
Accordingly, vector $\theta$  defined as $\theta:= [\rho,\lambda,\theta_{\kernel}]\in\Theta$ is the overall vector of hyperparameters to be determined where $\Theta$ denotes the space of feasible hyperparameters.
For estimating $\theta$, we employ a \emph{cross-validation} scheme endowed with a \emph{Bayesian optimization} heuristic \cite{srinivas2012information}.
In this regard, we split the index set of data $\{0,1,\ldots,\nD-1\}$ into two mutually disjoint subsets $\Ical_{\text{T}}$ and $\Ical_{\text{V}}$ respectively for \emph{training} and \emph{validation}. 
We define the model evaluation metric $v:\Theta\to\Rbb$ based on the prediction error of the validation data as 
\begin{equation}
	v(\theta) = \frac{1}{|\Ical_{\text{V}}|}\sum_{i\in\Ical_{\text{V}}}
	\big(
	y_{t_i}-{a}\Lu{t_i}(\vcf_{\rho})
	-\Lu{t_i}(\vch)
	\big)^2,
\end{equation}
where $({a},\vch)$ is the result of proposed identification scheme using the training data corresponding to indices $\Ical_{\text{T}}$, and also to the hyperparameter $\theta$. 
We then estimate the hyperparameters $\theta$ as $\hat{\theta}:=\argmin_{\theta\in\Theta} \!\ v(\theta)$.
Since the dependency of $v(\theta)$ on $\theta$ has a black-box oracle form, we need to employ Bayesian optimization algorithms such as \emph{GP-LCB} or similar alternatives \cite{srinivas2012information}. These heuristics are readily available in \MATLAB's \texttt{bayesopt} function.

\section{Further Internal Positivity Side-Information and Extensions} \label{sec:extensions}
In Section \ref{sec:positive_realization}, we employed positive system realization theory to formulate the identification problem with internal positivity side-information. 
The resulting optimization problem, \eqref{eqn:opt_h_rank}, is formulated using the fact that the transfer function of system $\Scal$ is in the following form
\begin{equation}\label{eqn:G_dom_sup}
\GS(z)=\FS(z)+\HS(z) = \frac{a}{1-\rho z^{-1}} + \HS(z).
\end{equation}
The transfer function $\GS$ has a dominant part $\FS(z) := a(1-\rho z^{-1})^{-1}$ with $\rho\in(0,1)$, and a suppressed part $\HS$.
Given that the impulse response of the system satisfies specific properties,
this formulation can be further extended to the following cases 
\begin{equation}\label{eqn:G_NUP}
\GS(z)=\FS(z)+\HS(z)  = \frac{\NS(z^{-1})}{(1-\rho z^{-1})^n} + \HS(z),
\end{equation}
and
\begin{equation}\label{eqn:G_SNP}
\GS(z)=\FS(z)+\HS(z)  = \frac{\NS(z^{-1})}{1-\rho^n z^{-n}} + \HS(z),
\end{equation}
where $\NS$ is a polynomial with degree less than $n$, i.e., $\NS\in\Rbb_{n-1}[z^{-1}]$.
Note that according to \cite[Theorem 9]{benvenuti2004tutorial}, the transfer function of a positive system with non-zero spectral radius is in form of \eqref{eqn:G_NUP} or \eqref{eqn:G_SNP}, which are generalized forms of \eqref{eqn:G_dom_sup}. 
These cases correspond to the situations where, in addition to internal positivity, we have further information on the dominant part of the impulse response of the system. 
In this section, we discuss these extensions.
\subsection{Non-simple Unique Dominant Pole}
We first discuss the extension which corresponds to \eqref{eqn:G_NUP}, i.e., the transfer function has a unique dominant pole with multiplicity $n$ where $n$ can be larger than one.
The mathematical proofs are omitted since they are similar to the ones given in Section \ref{sec:Towards_a_Tractable_Solution}. 

With respect to each $n\in\Nbb$ and $\rho\in(0,1)$, let $\Fscran$
be the following set of impulse responses 
\begin{equation}\label{eqn:Fscran}
	\begin{split}
		\Fscran &= \bigg\{ \vcf=(f_t)_{t=0}^{\infty}\in\ellone
		\!\ \Big|\!\  	
		\limOp_{t\to \infty} t^{-n+1}\rho^{-t}f_t >0,
		\\& \qquad\qquad
		(1-\rho z^{-1})^n \!\
		\sum_{t=0}^{\infty}f_t z^{-t} \in \Rbb_{n-1}[z^{-1}]
		\bigg\}.
	\end{split}
\end{equation}
Also, define the impulse response sets $\Pscran$ and $\PscrIn$ respectively as
\begin{equation}\label{eqn:Pscran}
\begin{split}
\Pscran &= \Big\{ \vcg = \vcf+\vch\in\underline{\Pscr} 
\!\ \Big|\!\  	
\vcf\in\Fscran, 
\\&\qquad\quad\ 
\vch=(h_t)_{t=0}^{\infty}\in\ell_{1},
\lim_{t\to \infty} \rho^{-t}h_t =0
\Big\}, 
\end{split}
\end{equation}
and $\PscrIn = \cup_{\rho\in(0,1)}\Pscran$.
According to \cite[Theorem 11]{benvenuti2004tutorial} and based on an argument similar to the proof of Corollary~\ref{cor:dom_pole_simple}, one can show that, for any $\rho\in(0,1)$, the impulse responses in $\Pscran$ are internally positive. 
Indeed, $\Pscran$ is exactly the set of impulse responses of positive systems with dominant pole structure as in \eqref{eqn:G_SNP}. 
Thus, to identify impulse response $\gS$ with internal positivity side-information in the sense that $\gS\in \Pscran$ we need to estimate $\vcf$ and $\vch$ with properties given in \eqref{eqn:Pscran}.
One can see that each $\vcf=(f_t)_{t=0}^{\infty}\in\Fscran$ is uniquely characterized in terms of real positive number $a$ and vector $\vca = [a_j]_{j=0}^{n-2}\in\Rbb^{n-2}$ as
\begin{equation}
	f_t = at^{n-1}\rho^t+\sumOp_{j=0}^{n - 2}a_jt^j\rho^t,\qquad \forall t\ge 0.
\end{equation} 
Subsequently, we reintroduce the empirical loss function $\Ecal_{\rho,n}:\Rbb\times\Rbb^{n-2}\times\Hk\to\Rbb_+$ as following
\begin{equation}\label{eqn:E_an}
\begin{split}
	&\Ecal_{\rho,n}(a,\vca,\vch)  := 
	\\&\quad
	\sum_{i=0}^{\nD - 1}
	\Big[
	y_{t_i}
	-
	\Lu{t_i}\Big(a\vcf_{\rho,n - 1}+\sum_{j=0}^{n -2}a_j\vcf_{\rho,j}\Big)
	-\Lu{t_i}(\vch)
	\Big]^2\!\!,
\end{split}
\end{equation}
where, for $j=0,\ldots,n-1$, the impulse response $\vcfaj$ is defined as 
$\vcfaj = (t^j\rho^t)_{t=0}^{\infty}$.
According to \eqref{eqn:E_an}, the identification problem \eqref{eqn:opt_h_rank} is updated to the following optimization problem
\begin{equation}\label{eqn:opt_h_rank_UNP} 
	\begin{array}{cl}
		\minOp_{\substack{a\in\Rbb,\vch\in\Hk\\\vca\in\Rbb^{n-1}}} & 
		\Ecal_{\rho,n}(a,\vca,\vch) + \lambda\!\ \|\vch\|_{\Hk}^2 
		+ \varepsilon\|\vca\|^2,\\
		\text{s.t.}
		&h_t+\rho^t\Big[at^{n-1}\!+\!\sumOp_{j=0}^{n - 2}a_jt^j\Big]\ge 0, \  \forall t \ge 0,\\
		&\rank(\Hankel(\vch))< \infty,\\
		&a\ge a_{\min},
	\end{array} 
\end{equation}
where $a_{\min}>0$ is a given lower-bound for $a$ and $\varepsilon>0$ is a regularization weight.
Based on the same line of argument as in Section \ref{sec:towards_tractable}, the next proposition presents a finite dimensional convex quadratic program equivalent to \eqref{eqn:opt_h_rank_UNP}.
Before proceeding to the proposition, we define
matrices $\mxB\in \Rbb^{\nD\times n}$ and 
$\mxC\in \Rbb^{(m+1)\times n}$ as follows
\begin{equation}\label{eqn:mxBC}
	\!\!\!\!\!\!
	\begin{array}{ll}
		\mxB(i,j)  = 
		\Lu{t_{i-1}}(\vcf_{\rho,j-1}), &  1\!\le\! i\!\le\! \nD,1\!\le\! j\!\le\! n, \\
		\mxC(i,j)  = (i-1)^{j-1}\rho^{i-1},& 1\!\le\! i\!\le\! m \! +\! 1, 1\!\le\! j\!\le\!n.
	\end{array}
	\!\!\!\!\!\!
\end{equation}
\begin{theorem}\label{prp:rep_posi_id_NUP}
	Let Assumptions \ref{ass:kernel_dominated}, \ref{ass:dom_pole_exciting} and \ref{ass:finite_input} hold and $\kernel$ be a finite Hankel rank kernel.
	Also, with respect to each $m\in\Zbb_+$, let $(a^\mt,\vca^\mt,\vcx^\mt)$ be the solution of following  convex quadratic program 
	\begin{equation}\label{eqn:min_rho_x_finite_dim_NUP}
		\!\!\!\!\!\!\!
		\begin{array}{cl}
			\!\!\!
			\minOp_{\substack{\vcx\in \Rbb^{\nD+m+1}\\
					a\in\Rbb,\vca\in\Rbb^{n-1}}}
			\!\!\!\!&\!
			\Big\|
			\vcy
			\!-\!
			\mxB\!
			\begin{bmatrix} \vca \\  a \end{bmatrix}
			\!\!-\!
			\begin{bmatrix}\mxO \!\!&\!\!  \mxL\end{bmatrix}
			\!\vcx
			\Big\|^2 
			\!+\!
			\lambda
			\vcx^\tr
			\!\!
			\begin{bmatrix}
				\mxO 		\!\!&\!\!	\mxL\\
				\mxL^\tr	\!\!&\!\!	\mxK\\
			\end{bmatrix}
			\vcx
			\!+\!
			\varepsilon\|\vca\|^2\!\!,
			\\
			\!\!\!
			\mathrm{s.t.}
			\!\!\!\!&\!
			\begin{bmatrix}
				\mxL^\tr \!\!&\!\! \mxK\\
			\end{bmatrix}
			\vcx+ \mxC\!
			\begin{bmatrix} \vca \\  a \end{bmatrix}\ge 0,
			\\\!\!\!&\!\!
			a\ge a_{\min},
		\end{array} 
		\!\!\!\!\!\!\!\!
		\!\!\!\!\!\!\!\!
	\end{equation}	
	and, the impulse response $\vch^\mt$ is defined according to \eqref{eqn:hm_rep}.
	Then, there exists $m^*$ such that $(a^\mt,\vca^\mt,\vch^\mt)$ is a solution of \eqref{eqn:opt_h_rank_UNP}, for any $m\ge m^*$.
	Moreover, for any $m_1,m_2\ge m^*$, we have 
	$(a^{\supscrpsm{m_1}},\vca^{\supscrpsm{m_1}},\vch^{\supscrpsm{m_1}})
	=
	(a^{\supscrpsm{m_2}},\vca^{\supscrpsm{m_2}},\vch^{\supscrpsm{m_2}})$. 
\end{theorem}
\subsection{Multiple Simple Dominant Poles}
In this section, we introduce the extension corresponding to the case in \eqref{eqn:G_SNP}, i.e., the dominant part of the transfer function has specially structured multiple simple dominant poles.

For any $n\in\Nbb$ and any $\rho\in(0,1)$, defined the impulse response sets $\Fscrna$, $\Pscrna$ and $\PscrnI$ respectively as
\begin{equation}\label{eqn:Pscrna}
	\begin{split}
		\Fscrna &= \bigg\{ \vcf=(f_t)_{t=0}^{\infty}\in\underline{\Pscr} 
		\!\ \Big|\!\  	
		\liminfOp_{t\to \infty} \rho^{-t}f_t >0,
		\\& \qquad\qquad
		(1-\rho^n z^{-n})
		\sum_{t=0}^{\infty}f_t z^{-t} \in \Rbb_{n-1}[z^{-1}]
		\bigg\}, \\
		\Pscrna &= \bigg\{ \vcg = \vcf+\vch\in\underline{\Pscr} 
		\!\ \Big|\!\  	
		\vcf\in\Fscrna, 
		\\&\qquad\qquad
		\vch=(h_t)_{t=0}^{\infty}\in\ell_{1},
		 \liminfOp_{t\to \infty} \rho^{-t}|h_t| =0
		\bigg\}, 
	\end{split}
\end{equation}
and $\PscrnI = \cup_{\rho\in(0,1)}\Pscrna$.
One can easily see that for $n=1$ these sets coincide with $\Pscra$ and $\PscrI$.
Due to  \cite[Theorem 12]{benvenuti2004tutorial} and by following the same lines of argument as in the proof of Corollary~\ref{cor:dom_pole_simple}, we can show that $\Pscrna$ contains internally positive impulse responses, for any $\rho\in(0,1)$.
Accordingly, the identification with internally positivity side-information in the sense that impulse response belongs to $\Pscrna$ translates to the estimation of $\vcf$ and $\vch$ with properties given in \eqref{eqn:Pscrna}.
Note that with respect to each $\vcf=(f_t)_{t=0}^{\infty}\in\Fscrna$, there exist real scalars $\ar_0,\ldots,\ar_{n-1}$ and $\ai_0,\ldots,\ai_{n-1}$ such that, for any $t\in\Zbb_+$, we have
\begin{equation}
\begin{split}
f_t &= 
\real\Big(\sum_{k=0}^{n-1} (\ar_k+\Jimage\ai_k)\rho^t\omega^{kt} \Big),
\\
0&=
\imag\Big(\sum_{k=0}^{n-1} (\ar_k+\Jimage\ai_k)\rho^t\omega^{kt} \Big),
\end{split}
\end{equation} 
where $\omega=\expe^{\frac{2\pi}{n}}$, i.e., $\vcf$ is uniquely characterized in terms of vectors $\vcar = [\ar_k]_{k=0}^{n-1}$ and 
$\vcai = [\ai_k]_{k=0}^{n-1}$.
Hence, we can reintroduce the empirical loss function 
$\Ecal_{\rho}^{\supscrpsm{n}}:\Rbb^{n}\times\Rbb^{n}\times\Hk\to\Rbb_+$ as  follows
\begin{equation*}
\begin{split}
	\Ecal_{\rho}^{\supscrpsm{n}}&(\vcar,\vcai,\vch)  
	:= 
	\\&
	\sum_{i=0}^{\nD - 1}
	\Big[
	y_{t_i}
	-
	\Lu{t_i}\Big(
	\sum_{k=0}^{n-1}( \ar_k\vcfark-\ai_k\vcfaik)
	\Big)
	-\Lu{t_i}(\vch)
	\Big]^2\!\!,
\end{split}
\end{equation*}
where $\vcfark$ and $\vcfaik$ 
are defined as
\begin{equation}
\vcfark = 
\big(\rho^t \!\ \real(\omega^{kt})\big)_{t=0}^{\infty}, 
\quad
\vcfaik = 
\big(\rho^t \!\ \imag(\omega^{kt})\big)_{t=0}^{\infty},
\end{equation}
for $k=0,\ldots,n-1$.
Therefore, the identification problem \eqref{eqn:opt_h_rank} is modified to
\begin{equation}\label{eqn:opt_h_rank_SNP} 
	\!\!\!
	\!\!\!\!\!\!\!\!\!
	\!\!\!\!\!\!\!\!\!
	\!\!\!\!\!\!\!\!\!
	\begin{array}{cl}
		\minOp_{\substack{\vcar,\vcai\in\Rbb^{n}\\\vch\in\Hk}} & 
		\Ecal_{\rho}^{\supscrpsm{n}}(\vcar,\vcai,\vch) 
		\!+\! \lambda\!\ \|\vch\|_{\Hk}^2 
		\!+\! \varepsilon\|\mxE\vcar\|^2
		\!+\! \varepsilon\|\mxE\vcai\|^2 \!\!,\\
		\text{s.t.}
		&h_t+\sumOp_{k=0}^{n-1}( \ar_k\vcfark-\ai_k\vcfaik)\ge 0, \quad  \forall t \ge 0,\\
		&\sumOp_{k=0}^{n-1}( \ai_k\vcfark+\ar_k\vcfaik)= 0, \quad   \forall t \ge 0,\\
		&\liminfOp_{t\to\infty} \rho^{-t}\sumOp_{k=0}^{n-1}( \ar_k\vcfark-\ai_k\vcfaik)\ge a_{\min},\\
				&\rank(\Hankel(\vch))< \infty,\\
	\end{array} 
	\!\!\!\!\!\!\!\!\!
	\!\!\!\!\!\!\!\!\!
	\!\!\!\!\!\!\!\!\!
\end{equation}
where  ${a}_{\min}>0$ is a small positive real scalar, 
$\varepsilon>0$ is a regularization weight and $\mxE\in\Rbb^{n\times n}$ is defined as  $\mxE = \diag(0,1,1,\ldots,1)$.
By an argument similar to Section \ref{sec:towards_tractable}, we can derive an equivalent finite dimensional convex quadratic program for \eqref{eqn:opt_h_rank_SNP}.
To this end, define matrices $\mxV_m\in\Rbb^{m\times n}$, $\mxBr\in\Rbb^{\nD\times n}$, and,  $\mxBi\in\Rbb^{\nD\times n}$, respectively as
\begin{equation}\label{eqn:mxVm}
\begin{array}{ll}
\mxV_m(i,j)  = \omega^{(i-1)(j-1)}, &  1\!\le\! i\!\le\! m,1\!\le\! j\!\le\! n,\\
\mxBr(i,j) = \Lu{t_{i-1}}\!(\vcf_{\rho,j-1}^{\supscrpsm{\mathrm{r}}}), &  1\!\le\! i\!\le\! \nD,1\!\le\! j\!\le\! n,\\
\mxBi(i,j) = \Lu{t_{i-1}}\!(\vcf_{\rho,j-1}^{\supscrpsm{\mathrm{i}}}), &  1\!\le\! i\!\le\! \nD,1\!\le\! j\!\le\! n,\\
\end{array}
\end{equation}
for  $m \in \Zbb_+ \cup \{\infty\}$.
Also, let $\mxVr_m$, $\mxVi_m$ and $\mxD_m$ be defined respectively as $\real(\mxV_m)$, $\imag(\mxV_m)$, and $\diag(1,\rho,\ldots,\rho^{m-1})$.
One can see that the first constraint in \eqref{eqn:opt_h_rank_SNP} is equivalent to 
\begin{equation}
\vch + \mxD_{\infty}\mxVr_{\infty}\!\ \vcar-\mxD_{\infty}\mxVi_{\infty}\!\ \vcai\ge 0.	
\end{equation}
The second constraint in \eqref{eqn:opt_h_rank_SNP} is
\begin{equation}\label{eqn:DVrai_DViar_3rd_const}
	\mxD_{\infty}\mxVr_{\infty}\!\ \vcai+\mxD_{\infty}\mxVi_{\infty}\!\ \vcar= 0,
\end{equation}
which  implies that  
$\mxVr_{\infty} \vcai+\mxVi_{\infty}\vcar=0$ due to $\rho>0$.
As $\omega^n=1$, we know that $\mxV_m$ is an $n$-periodic Vandermonde matrix.
Therefore, \eqref{eqn:DVrai_DViar_3rd_const} is equivalent to $\mxVr_{n} \vcai+\mxVi_{n}\vcar=0$.
Similarly, one can show that the third constraint in \eqref{eqn:opt_h_rank_SNP}
is equivalent to $\mxVr_{n}\vcar+\mxVi_{n}\vcai\ge a_{\min}1_n$.
Based on the discussion above and similar to those in Section \ref{sec:Towards_a_Tractable_Solution},  the next proposition presents the
finite-dimensional convex quadratic program equivalent to \eqref{eqn:opt_h_rank_SNP}.
\begin{theorem}\label{prp:rep_posi_id_SNP}
Let the assumptions of Proposition \ref{prp:rep_posi_id_NUP} hold.
For any $m\in\Zbb_+$, let 
$(\vca^{\supscrpsm{\mathrm{r},m}},\vca^{\supscrpsm{\mathrm{i},m}},\vcx^\mt)$ be the solution of following  convex quadratic program 
\begin{equation}\label{eqn:min_rho_x_finite_dim_SNP}
\!\!\!\!\!\!\!
\begin{array}{cl}
	\!\!\!
	\minOp_{\substack{\vcx\in \Rbb^{\nD+m+1}\\
			\vcar\!,\!\ \vcai\in\Rbb^{n}}}
	\!\!\!\!&\!
	\Big\|
	\vcy
	\!-\!
	\mxBr\vcar 
	\!-\!
	\mxBi\vcai
	\!-\!
	\begin{bmatrix}\mxO \!\!&\!\!  \mxL\end{bmatrix}
	\!\vcx
	\Big\|^2 
	\!\!\!+\!
	\lambda
	\vcx^\tr
	\!\!
	\begin{bmatrix}
		\mxO 		\!\!&\!\!	\mxL\\
		\mxL^\tr	\!\!&\!\!	\mxK\\
	\end{bmatrix}
	\!
	\vcx
	\\&\qquad
	+\!\ \varepsilon \|\mxE\vcar\|^2+\varepsilon\|\mxE\vcai\|^2\!\!,
	\\
	\!\!\!
	\mathrm{s.t.}
	\!\!\!\!&\!
	\begin{bmatrix}
		\mxL^\tr \!\!&\!\! \mxK\\
	\end{bmatrix}\!
	\vcx
	+  
	\mxD_{m +1}
	\!
	\begin{bmatrix} \mxVr_{m+1} \!\!&\!\!-\mxVi_{m+1}\!\end{bmatrix}
	\!\!
	\begin{bmatrix} \vcar \\  \vcai \end{bmatrix}\!\!\ge\! 0,
	\\\!\!\!&\!\!
	\mxVr_{n} \vcai+\mxVi_{n}\vcar=0,
	\\\!\!\!&\!\!
	\mxVr_{n}\vcar+\mxVi_{n}\vcai\ge a_{\min}1_n,
\end{array} 
\!\!\!\!\!\!\!\!
\!\!\!\!
\end{equation}	
and, define the  impulse response $\vch^\mt$ by \eqref{eqn:hm_rep}.
Then, there exists $m^*$ such that 
$(\vca^{\supscrpsm{\mathrm{r},m}},\vca^{\supscrpsm{\mathrm{i},m}},\vcx^\mt)$ 
is a solution of \eqref{eqn:opt_h_rank_UNP}, for each $m\ge m^*$.
Moreover, for any $m_1,m_2\ge m^*$, we have 
$(\vca^{\supscrpsm{r,m_1}},\vca^{\supscrpsm{i,m_1}},\vch^{\supscrpsm{m_1}})
=
(\vca^{\supscrpsm{r,m_2}},\vca^{\supscrpsm{i,m_2}},\vch^{\supscrpsm{m_2}})$. 
\end{theorem}
\subsection{Zero Spectral Radius}
According to \cite[Theorem 8]{benvenuti2004tutorial}, when the transfer function $\GS$ has zero spectral radius, $r(\GS)=0$, the impulse response $\gS$ is internally positive if $\gS$ is non-negative. 
Based on the proof of \cite[Theorem 8]{benvenuti2004tutorial}, we know that $\gS$ belongs to $c_{00}$, the space of finitely non-zero impulse responses, i.e., there exists $\nG\in\Zbb_+$ such that $\GS(z)=\sum_{t=0}^{\nG-1}g_tz^{-t}$ and $g_t=0$, for all $t\ge \nG$.
Comparing to the other cases of internal positivity side-information, the case of zero spectral radius provides the weakest information.
Indeed, this knowledge only says the $\gS$ is a non-negative and finitely non-zero sequence and provides no further information about the behavior of the impulse response of system $\Scal$.

Based on the discussion above, the identification problem with the internal positivity side-information and the extra knowledge $r(\GS)=0$ can be easily formulated as following
\begin{equation}\label{eqn:opt_h_ZSR}
	\begin{array}{cl}
		\minOp_{\vcg\in\Hk} & 
		 \sum_{i=0}^{\nD - 1}
		 \big(
		 y_{t_i}
		 -\Lu{t_i}(\vcg)
		 \big)^2 + \lambda\!\ \|\vcg\|_{\Hk}^2,\\
		\text{s.t.}
		&g_t\ge 0, \quad \forall t=0,1,\ldots,\nG-1,\\
	\end{array} 
\end{equation}
where $\kernel$ is a kernel zero on $\Zbb_+^2\backslash\{0,\ldots,\nG-1\}^2$.
One can show that for $\vcg^* =(g_t^*)_{t=0}^{\infty}$,  the solution of \eqref{eqn:opt_h_ZSR}, and for $t=0,1,\ldots,\nG$, we have 
\begin{equation}\label{eqn:g_rep}
	g_t = 
	\sum_{i=0}^{\nD} x_i \Lu{t_i}(\kernel_t) + 
	\sum_{s=0}^{\nG-1}  x_{\nD+s} \kernel_s(t), 
\end{equation}	
where $\vcx=[x_0,\ldots,x_{\nD+\nG-1}]^\tr$ is the solution of following \emph{convex quadratic program}
\begin{equation}\label{eqn:min_x_ZSR}
	\begin{array}{cl}
		\minOp_{\vcx\in \Rbb^{\nD+\nG}}
		&
		\big\|
		\vcy
		-\begin{bmatrix}\mxO \!\!&\!\!  \mxL\end{bmatrix}
		\vcx
		\big\|^2 
		+
		\lambda
		\vcx^\tr
		\begin{bmatrix}
			\mxO 		\!\!&\!\!	\mxL\\
			\mxL^\tr	\!\!&\!\!	\mxK\\
		\end{bmatrix}
		\vcx,
		\\
		\mathrm{s.t.}
		&
		\begin{bmatrix}
			\mxL^\tr \!\!&\!\! \mxK\\
		\end{bmatrix}
		\vcx\ge 0.
	\end{array} 
\end{equation}	
Also, one can solve \eqref{eqn:opt_h_ZSR} directly for $[g_0,\ldots,g_{\nG-1}]^\tr\in\Rbb^{\nG}$ at the cost of a matrix inversion which can be  both imprecise and computationally demanding, especially when $\nG$ is large.
\section{Numerical Experiments}\label{sec:numerics}
In this section, we provide numerical and experimental examples to verify the efficacy and performance of the proposed method for impulse response identification with internal positivity side-information.
The first example concerns the impact of incorporating internal positivity side-information on the estimation quality and providing a comparative analysis for the proposed identification scheme through a Monte Carlo analysis. The second example concerns the efficacy of the proposed identification scheme on a set of data collected from an experimental heating system.

\subsection{Monte Carlo Experiment}\label{ssec:MC}

Consider system $\Scal$ described with the impulse response $\gS=(\gtS_t)_{t=0}^{\infty}$ 
defined as
\begin{equation}
	\gtS_t = {{\rho}}^t(1+\beta^t\cos(2\pi\omega t)), \qquad \forall\!\ t\in\Zbb_+,
\end{equation}
where ${{\rho}}$ and $\beta$ are real scalars in $(0,1)$, and $\omega$ is an irrational real number in $(0,1)$.
One can easily see that $\gS$ is a non-negative impulse response with transfer function
\begin{equation}
	\GS(z) \!=\! \frac{1}{1\!-\!{{\rho}} z^{-1}}+
	\frac{1-{{\rho}}\beta\cos w \ z^{-1}}
	{1-2{{\rho}} \beta\cos w \ z^{-1}+{{\rho}}^2 \beta^2 z^{-2}}.
\end{equation}
Therefore, according to Corollary \ref{cor:dom_pole_simple}, we know that $\gS$ is internally positive.

\subsubsection{Simulation Configuration}
In this numerical experiment, we set ${{\rho}}=0.98$, $\beta=0.92$, and $\omega =\frac1{10}\pi^2$.
Using \MATLAB's \texttt{idinput} function, we generate a set of $120$ random binary input signals, 	  each with length of $\nD=200$.
The system is initially at rest. The input signals are applied to the system and the corresponding noiseless output is obtained. 
We consider three  signal-to-noise ratio (SNR)  levels of $10$ dB, $20$ dB, and $30$ dB. 
With respect to each of these SNR levels and each output signal, we generate a zero-mean white Gaussian signal as the additive measurement uncertainty. 
The resulting noisy output is measured at time instants $t_i=i$, for $i=0,1,\ldots,199$.
Accordingly, with respect to each of the mentioned SNR levels, we have $120$ sets of input-output data  as follows4
\begin{equation*}
\begin{split}
\Dscr_i^{\supscrpsm{10\text{dB}}}
&\!=\!
\big\{(u_s^{\supscrpsm{i}},y_s^{\supscrpsm{i,10\text{dB}}})\big|s \!=\! 0,\ldots,199\big\},\ \ i \!=\! 1,\ldots,120,	
\\
\Dscr_i^{\supscrpsm{20\text{dB}}}
&\!=\!
\big\{(u_s^{\supscrpsm{i}},y_s^{\supscrpsm{i,20\text{dB}}})\big|s \!=\! 0,\ldots,199\big\},\ \ i \!=\!1,\ldots,120,	
\\
\Dscr_i^{\supscrpsm{30\text{dB}}}
&\!=\!
\big\{(u_s^{\supscrpsm{i}},y_s^{\supscrpsm{i,30\text{dB}}})\big|s \!=\! 0,\ldots,199\big\},\ \ i \!=\!1,\ldots,120,	
\end{split}
\end{equation*}	
where the superscript refers to the SNR level in the corresponding set of data.
\subsubsection{Comparison Methods}
For estimating the impulse response of system, we utilize the  input-output data sets and the following identification methods:
\begin{itemize}
\item[A.] The first method is based on the subspace approach implemented by \MATLAB's \texttt{n4sid} and using the true order of system $\Scal$. Once we  obtain an initial estimation $\tilde{\vcg}^{\supscrpsm{1}}$, the result is projected on the positive orthant by $\hat{\vcg}^{\supscrpsm{1}}=\max(\tilde{\vcg}^{\supscrpsm{1}},0)$, where the max-operation is performed coordinate-wise. 
\item[B.]In the second method, we employ the least squares approach, and then, similar method A, the projection on the positive orthant is applied to estimate a non-negative finite impulse response. More precisely, the estimation $\hat{\vcg}^{\supscrpsm{2}}$ is obtained by $\hat{\vcg}^{\supscrpsm{2}}=\max(\tilde{\vcg}^{\supscrpsm{2}},0)$, where 
\begin{equation}
\tilde{\vcg}^{\supscrpsm{2}}:=
\argmin_{\vcg\in\Rbb^{\nG}}
\|\mxT_{\vcu}\vcg-\vcy\|^2,	
\end{equation}
and, vector $\vcy$ and  Toeplitz matrix $\mxT_{\vcu}$ are respectively defined in \eqref{eqn:vec_y_b} and Remark \ref{rem:OLK_at_rest}.
\item[C.] This method is based on a constrained least squares approach, where the external positivity feature is enforced by setting the feasible set to the positive orthant. In other words, the impulse response is estimated as  
\begin{equation}
\hat{\vcg}^{\supscrpsm{3}}:=
\argmin_{\vcg\in\Rbb_+^{\nG}}\|\mxT_{\vcu}\vcg-\vcy\|^2.	
\end{equation}
\item[D.]\&\!\!\,\,\,\,\,E. The fourth and fifth methods are respectively similar to the second and third approaches, but with an additional kernel-based regularization term included in the corresponding optimization problems. For method D, one can employ \MATLAB's \texttt{impulseest} function and then taking the positive part of the resulting FIR.
Also, method E essentially corresponds to  \eqref{eqn:opt_h_ZSR}, or equivalently \eqref{eqn:min_x_ZSR}.
\item[F.] The sixth method is a Bayesian FIR estimation scheme for externally positive systems \cite{zheng2021bayesian,zheng2018positive}. This scheme is based on maximum a posteriori estimation, where the employed prior is a maximum entropy distribution with support on positive orthant and kernel-based covariance.
\item[G.] The last method is the scheme proposed in this paper and summarized in Algorithm \ref{alg:PosiID} (see Section \ref{sec:alg}).
\end{itemize}
One should note that in all of the mentioned methods, the resulting impulse responses are non-negative.
In order to have a fair comparison,  in the kernel-based methods D, E, F, and G, the same kernel type \eqref{eqn:kernelDC} is employed.

\subsubsection{Evaluation Metrics and Results}
For evaluating the performances of these methods, we compare the resulting bias-variance trade-offs, as shown in Table \ref{tbl:example_1_bias_variance}. 
\begin{table}[t]
\renewcommand{\arraystretch}{1.3} 
\centerline{
\begin{tabular}{c|lccccccc}
\toprule
\hline
&\!\!\!Method\!\!\!
& A & B & C & D & E & F & G  \\ 
\cline{2-9}
\multirow{3}{*}{\rotatebox{90}{\footnotesize{10 dB}}}
&\textbf{$\!\!\!\text{Bias}$}$(\hat{\vcg})\!\!\!\! $
&2.14 &2.86  &1.60  &0.80 &0.77 &0.62 &\cmb{0.27}\\
&\textbf{$\!\!\!\text{Var}$}$(\hat{\vcg})\!\!\!\! $
&6.54 &34.7  &11.8  &1.77  &1.74  &1.74&\cmb{0.51}\\
&\textbf{$\!\!\!\text{MSE}$}$(\hat{\vcg})\!\!\!\! $
&11.1 &42.9  &14.8  &2.40  &2.34  &2.06&\cmb{0.58}\\
\hline
\multirow{3}{*}{\rotatebox{90}{\footnotesize{20 dB}}}
&\textbf{$\!\!\!\text{Bias}$}$(\hat{\vcg})\!\!\!\! $
&2.01 &0.78 &0.48 &0.47 &0.38 &0.37 &\cmb{0.089}\\
&\textbf{$\!\!\!\text{Var}$}$(\hat{\vcg})\!\!\!\! $
&1.26 &3.90 &1.69 &0.99 &0.85 &0.84 &\cmb{0.114}\\
&\textbf{$\!\!\!\text{MSE}$}$(\hat{\vcg})\!\!\!\! $
&5.32 &4.50 &1.93 &1.22 &0.99 &0.97 &\cmb{0.122}\\
\hline
\multirow{3}{*}{\rotatebox{90}{\footnotesize{30 dB}}}
&\textbf{$\!\!\!\text{Bias}$}$(\hat{\vcg})\!\!\!\! $
&2.01 &0.27 &0.23 &0.22 &0.24 &0.25&\cmb{0.031}\\
&\textbf{$\!\!\!\text{Var}$}$(\hat{\vcg})\!\!\!\! $
&0.17 &0.58 &0.32 &0.33 &0.55 &0.54 &\cmb{0.015}\\
&\textbf{$\!\!\!\text{MSE}$}$(\hat{\vcg})\!\!\!\! $
&4.20 &0.65 &0.37 &0.38 &0.61 &0.60 &\cmb{0.016}\\
\hline
\bottomrule
\end{tabular}
}
\caption{The bias, variance and MSE resulting from the identification methods listed in Section \ref{ssec:MC}. The last column corresponds to the proposed approach which integrates internal positivity.
}
\label{tbl:example_1_bias_variance}
\end{table}
Moreover, for further quantitative comparison of the estimated impulse responses, we use \emph{coefficient of determination}, or \emph{R-squared}, which is denoted by $\mathrm{fit}$ and defined as
\begin{equation}\label{eqn:R2}
	\mathrm{fit}(\hat{\vcg}) = 100 \times \left(1-\frac{\|\hat{\vcg}-\gS\|_2}{\|\gS\|_2}\right),
\end{equation}
where $\hat{\vcg}$ is the estimated impulse response. 
Figure \ref{fig:exm1_box_plot} compares the resulting quality of fit for the different SNR levels.
\begin{figure}[t]
	\begin{center}
		\includegraphics[width=0.35\textwidth]{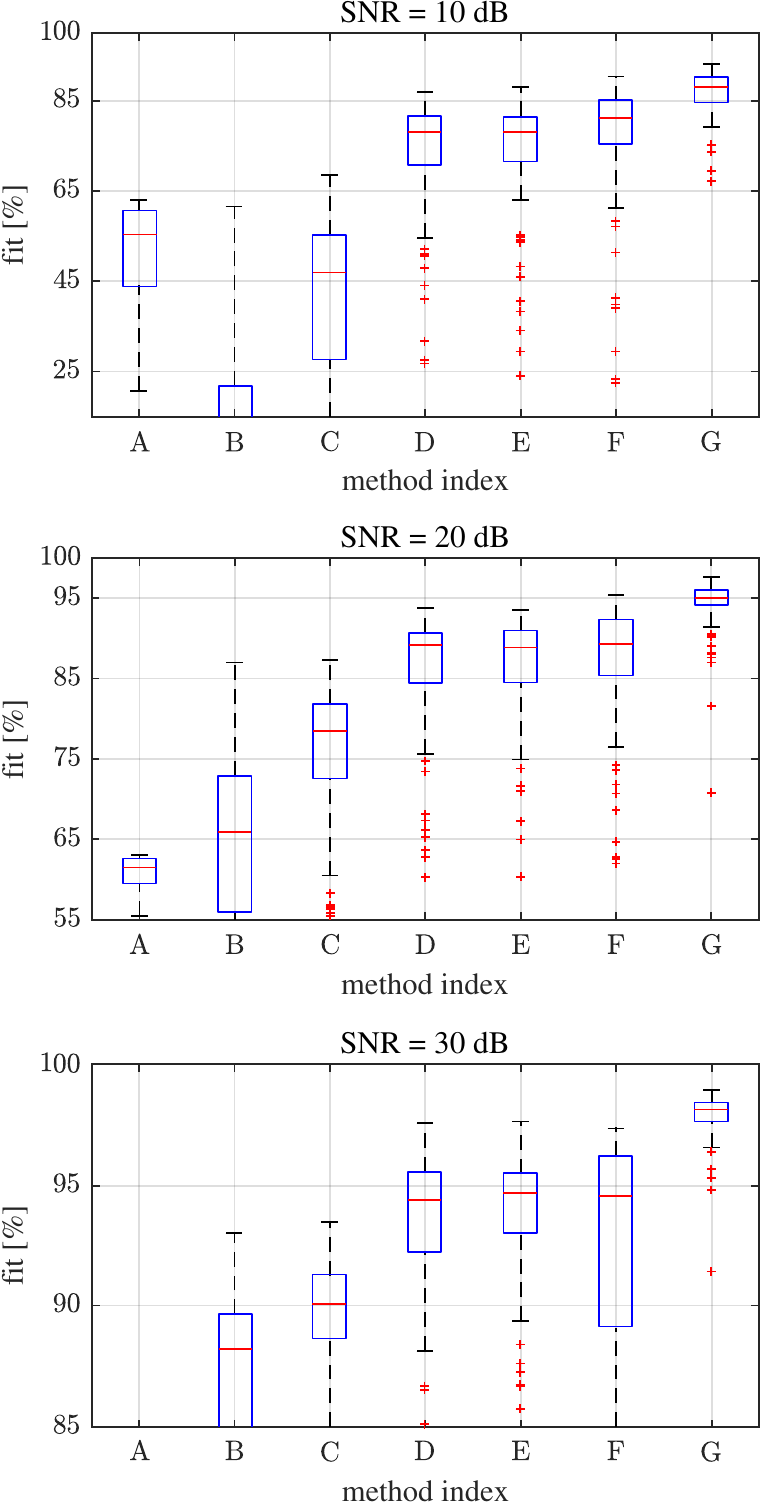}
	\end{center}
	\caption{Box-plots of the R-squared  metric for the estimation results of the different methods and SNR levels.
	All fits for method A are below $85$\%.	
	}
	\label{fig:exm1_box_plot}
\end{figure}

\subsubsection{Discussion}
We have several observations from the Monte Carlo numerical experiment.
As shown in Figure \ref{fig:exm1_box_plot}, all methods are outperformed by the proposed identification scheme, which maximally incorporates the internal positivity side-information. 
Indeed, the side-information helps excluding spurious model candidates and subsequently increases the accuracy of the estimation. The bias-variance results presented in Table \ref{tbl:example_1_bias_variance} confirm this fact.
While each of these methods estimates a non-negative impulse response and partially integrates positivity side-information, we can see that the level of integrated side-information is 
less than that of the proposed method.
The proposed approach incorporates this information in the model
maximally.
Comparing methods B and C, one can see the former one is a two-step procedure where the estimation is performed in the first step and non-negativity of the impulse response is obtained in the second step, while the latter approach is a single-step procedure that considers impulse response non-negativity during the estimation. 
On the other hand, according to the fitting results shown in Figure \ref{fig:exm1_box_plot},  C performs better than method B. For methods D and E, we have similar arguments.
This observation highlights the importance of jointly considering the positivity with the impulse response estimation, as done by the proposed method.
Method A knows the actual order of the system. However, according to the results presented in Figure \ref{fig:exm1_box_plot} and Table \ref{tbl:example_1_bias_variance}, one can see that positivity is a more advantageous and stronger side-information for impulse response estimation, especially when it is incorporated with its maximum strength like in the proposed scheme.  
Finally, one can see that the kernel-based methods D, E, F, and G have better estimation performance comparing to the methods A, B, and C, which is expected \cite{pillonetto2014kernel}.  

\begin{figure}[t]
	\begin{center}
		\includegraphics[width=0.485\textwidth]{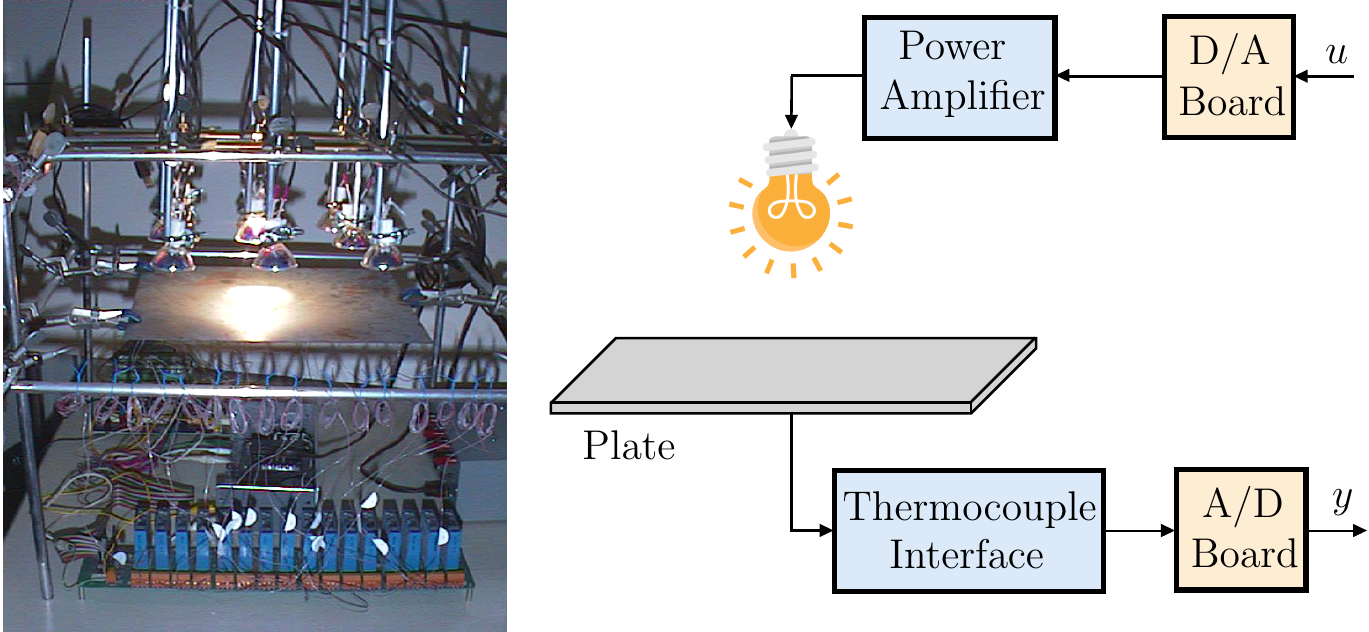}
	\end{center}
	\caption{The experimental system (left) and the corresponding block-diagram schematic (right).}
	\label{fig:exm2_photo_scheme}
\end{figure}

\subsection{Heating System Experiment}
\label{ssec:heating}
In this example, we verify the efficacy of the proposed identification scheme on a set of data collected from an experimental nonlinear heating system \cite{dullerud1996sampled}.

Figure \ref{fig:exm2_photo_scheme} shows the experiment configuration and the corresponding control and measurement schematic. In this experiment, a metal plate is heated up by a $300$-watt Halogen lamp mounted almost $5$cm above the center of the plate. On the other side of the plate, a thermocouple is placed, measuring the temperature. The thermocouple is connected to a computer via an A/D board for sampling and recording the temperature measurements. The lamp is supplied by a thyristor-based power amplifier driven by a D/A board and controlled by a computer. 
The sampling time for control and data acquisition is $T_{\text{s}}=2$s.
Accordingly, we have a nonlinear discrete-time system from the input of D/A board to the output of A/D board, which is intuitively close to a linear positive system.
The nonlinearity of the system is mainly due to the power amplifier \cite{dullerud1996sampled}. 
Moreover, the system is subject to delay  and   disturbances from the ambient.
To make the identification problem more challenging, we disregard the nonlinearity and external disturbances issues. 

The system is actuated by a piece-wise constant input and the output of system is measured for $801$ samples.  The collected input-output data is shown in Figure \ref{fig:exm2_data}, which is also available in the DAISY database \cite{de1997daisy}. 
\begin{figure}[b]
	\begin{center}
		\includegraphics[width=0.485\textwidth]{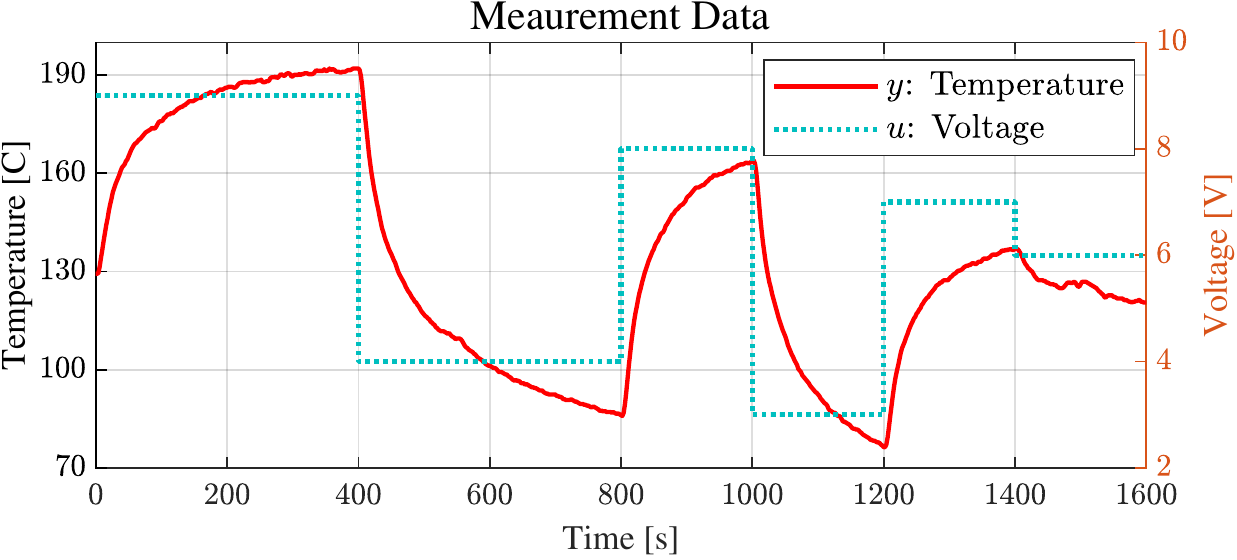}
	\end{center}
	\caption{The figure demonstrates the measurement data corresponding to the experiments.}
	\label{fig:exm2_data}
\end{figure}
We employ this data for identifying the system using the identification methods discussed in  Section \ref{ssec:MC}, and then, compare the results. 
Since the lamp has failed near to the end of experiment and the tail of data has less fidelity, we discard the last $101$ samples ($200$ seconds). 
We split the data into a \emph{training set}, to be used for identification, and, a \emph{test set}, which is the base for comparing the identified models. 
The samples training set contains the first $500$ measurement samples, and the next $200$ data points belong to the test set.
The quality of estimated model is evaluated based on the \emph{R-squared} metric, which assess the prediction precision on the test data, and, defined as follows
\begin{equation}\label{eqn:R2_y}
	\mathrm{fit}(\hat{\vcg}) = 100 \times \left(
	1-
	\Bigg[\frac
	{\sum_{501\le i\le 700}(y_{t_i}-\hat{y}_{t_i})^2}
	{\sum_{501\le i\le 700}(y_{t_i}-\ol{y})^2}
	\Bigg]^{\frac{1}{2}}
	\right),
\end{equation}
where $\hat{y}_s$ denotes the predicted output for time instant $s$ and $\ol{y}$ is the average of output measurements in the test set.
In \cite{dullerud1996sampled}, a Hammerstein model is derived where the static nonlinear block is a sinusoidal map derived by curve-fitting, and the linear block is an ARX model with estimated coefficients.  We denote this method by N+ARX.
In the kernel-based methods D, E, F, and G, we have employed the TC kernel \eqref{eqn:kernelTC}. 
Also, for the methods which are estimating an FIR, we have set $\nG=200$. 
We evaluate the R-squared metric on the test data for this model and the ones estimated by the above methods. 
Table \ref{tbl:example_2_R2_fit} reports fitting results where one can see that the proposed method provides more accurate fit.
\begin{table}[b]
	\renewcommand{\arraystretch}{1.5} 
	\centerline{
		\begin{tabular}{lcccccccc}
			\toprule
			\hline
			\!\!\textbf{Method}\!\!
			&N+ARX & A & B & C & D & E & F & G \\
			\cline{1-9}
			\!\!\textbf{Fit [\%]}\!\!
			&48.5 &83.4 &81.2 &80.4 &81.7 &89.7 &85.8 &\textbf{92.2}\\
			\hline
			\bottomrule
		\end{tabular}
	}
	\caption{R-squared metric evaluations of test data for different identification methods.}
	\label{tbl:example_2_R2_fit}
\end{table}
This is also confirmed by Figure \ref{fig:exm2_test_data} which compares the test data with the output signals predicted by methods N+ARX, E, F, and G.
It seems that for obtaining models with more accurate predictions, one should identify a nonlinear dynamics.
\begin{figure}[t]
	\begin{center}
		\includegraphics[width=0.43\textwidth]{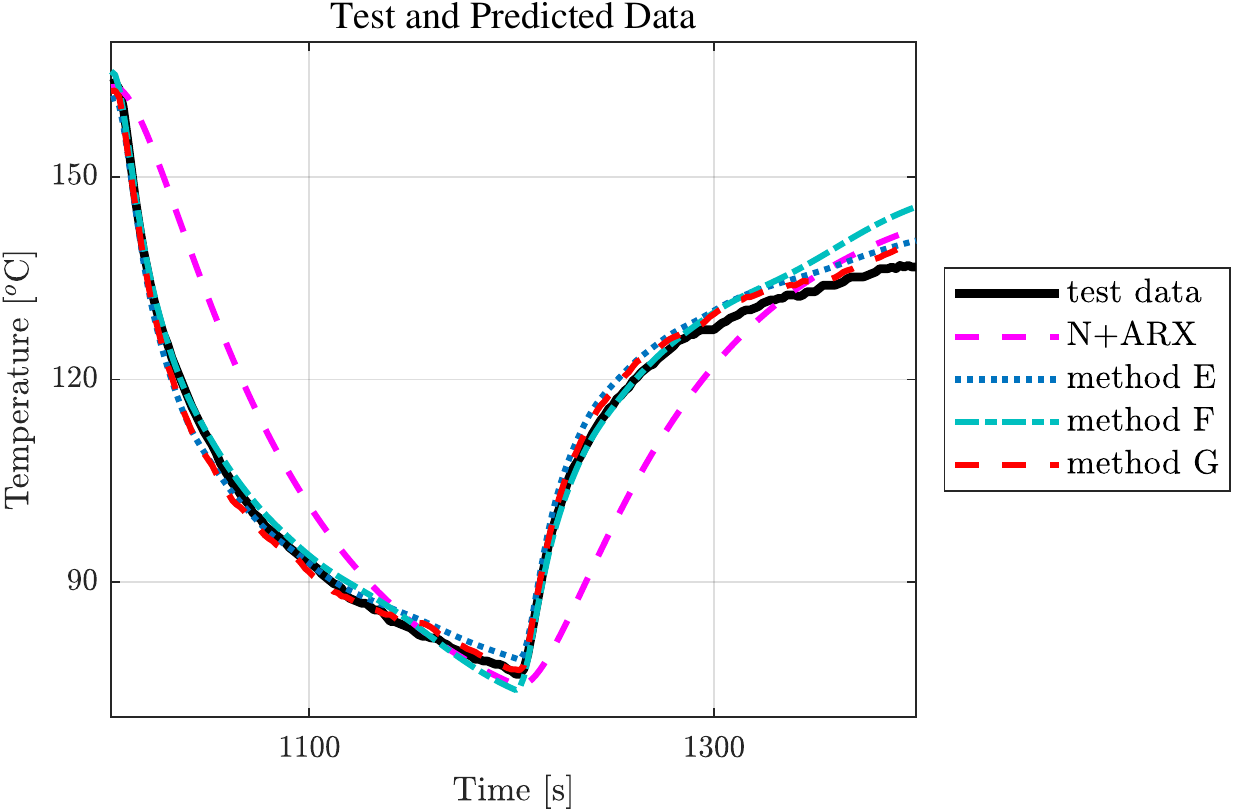}
	\end{center}
	\caption{The figure compares the test measurement data and the predicted values.}
	\label{fig:exm2_test_data}
\end{figure}

\section{Conclusion}\label{sec:conclusion}
In this paper, we have considered the problem of impulse response identification when side-information is available on the internal positivity of the system. We have employed the realization theory of positive systems to introduce the identification scheme in which the positivity side-information is integrated into the identified model. The resulting formulation is in the form of a constrained optimization over a reproducing kernel Hilbert space endowed with a stable kernel where the constraints are suitably designed to incorporate the positivity side-information in the solution. We have borrowed techniques and tools from optimization theory in normed spaces to derive an equivalent finite-dimensional convex quadratic program. This gives a computationally tractable identification scheme that incorporates the internal positivity side-information and has the well-known advantageous features of kernel-based methods. We have performed a Monte Carlo numerical experiment to compare the performance of proposed approach with FIR identification methods considering only the external positivity feature. This has empirically studied the impact of integrating positivity side-information in terms of estimation bias, variance, and mean squared error. The results show that the proposed identification approach, which integrates internal positivity, outperforms the schemes considering only external positivity. We have observed that incorporating internal positivity side-information reduces the estimation bias and variance. This observation is expected since FIR external positivity implies the weakest form of information about an internally positive system and fails to exploit the complete information of internal positivity. We have further evaluated the effectiveness of the proposed identification scheme using data from a heating system experiment.
\appendix
\section{Appendix} 
\subsection{Proof of Corollary \ref{cor:dom_pole_simple}} \label{sec:pf_dom_pole_simple}
Let $\gS\in\Pscr_{\!{{\rho}}}$	 and $\GS$ be the corresponding transfer function.
We know that $\gS$ is a non-negative impulse response which satisfies \eqref{eqn:rank_Hankel_gS}.
Accordingly, due to Theorem \ref{thm:kronecker}, there exist $\nx\in\Nbb$,  $\mxA\in\Rbb^{\nx\times \nx}$, $\vcb\in\Rbb^{\nx}$, $\vcc\in\Rbb^{1\times \nx}$ and $d\in\Rbb$ such that $\GS(z)=\vcc(z\eye-\mxA)^{-1}\vcb+d$.
Note that since $d = \gtS_0$, we know that $d\ge 0$.
Consider impulse response $\vcg=(g_t)_{t=0}^{\infty}$ where $g_0=0$ and $g_t=\gtS_t$, for $t\ge 1$, and let $G$ be the transfer function corresponding to $\vcg$. 
One can easily see that $\vcg$ is non-negative, and also, we have $G(z) = \GS(z)-d = \vcc(z\eye-\mxA)^{-1}\vcb$
which is a strictly proper rational transfer function. 
Since there exists ${a}>0$ such that $\lim_{t\to \infty} {{\rho}}^{-t}\gtS_t = {a}$, we know that $\lim_{t\to \infty} {{\rho}}^{-t}(g_t-{a}{{\rho}}^{t}) = 0$.
Therefore, the spectral radius of the rational transfer function
$G(z)-{a}(1-{{\rho}} z^{-1})^{-1}$ is less than ${{\rho}}$, and consequently, ${{\rho}}$ is the unique dominant pole of $G(z)$. 
Hence, according to Theorem \ref{thm:pos_real}, $G(z)$ admits a positive realization, i.e., there exist 
$\Mx\in\Nbb$,  $\mxA_+\in\Rbb_+^{\Mx\times\Mx}$, $\vcb_+\in\Rbb_+^{\Mx}$ and $\vcc_+\in\Rbb_+^{1\times\Mx}$ such that  $G(z)=\vcc_+(z\eye-\mxA_+)^{-1}\vcb_+$.
Therefore, we have $\GS(z)=\vcc_+(z\eye-\mxA_+)^{-1}\vcb_+ + d$ which says that $\GS$ has a positive realization due to $d\ge 0$.
Accordingly, $\gS$ is internally positive and $\gS\in\Pscr$, i.e., $\Pscr_{\!{{\rho}}}\subset\Pscr$.
From this result and the definition of $\Pscr_{\!(0,1)}$, the last claim is directly implied.
\subsection{Proof of Theorem \ref{thm:denseness_IP_EP}} \label{sec:pf_denseness_IP_EP}
Let $\varepsilon>0$ and $\vcg\in\Pscr$ with transfer function $G$.
Since $\Pscr\subset\ellone$, each element of $\Pscr$ is a stable impulse response, and therefore, we have $r(G)<1$. 
Let ${{\rho}}$ and ${a}$ be positive real scalars such that ${{\rho}}\in (r(G),1)$ and ${a}<(1-{{\rho}})\varepsilon$.
Consider an impulse response 
$\geps=(\gteps_t)_{t=0}^\infty$ with transfer function $\Geps$ where, for any $t\in\Zbb_+$, $g_t$ is defined as
$\gteps_t = g_t + {a}{{\rho}}^t$.
For any $t\in\Zbb_+$, one has  $g_t\ge 0$, and since ${a},{{\rho}}>0$, it follows that $\gteps_t\ge 0$, i.e., $\geps$ is a non-negative impulse response. 
Moreover, one can easily see that 
\begin{equation}\label{eqn:pf_01}
	\Geps(z) = G(z) + \frac{{a}}{1-{{\rho}} z^{-1}}.	
\end{equation}
Since $\vcg$ belongs to $\Pscr$, we know that $G$ is a rational function.	
Accordingly, due to \eqref{eqn:pf_01} and Theorem \ref{thm:kronecker}, it follows that \eqref{eqn:rank_Hankel_gS} holds for $\geps$. Moreover, ${{\rho}}>r(G)$ implies that 
$\lim_{t\to\infty}{{\rho}}^{-t}g_t=0$. Subsequently, one has $\lim_{t\to\infty}{{\rho}}^{-t}\gteps_t={a}$, and therefore, $\geps\in\Pscr_{\!{{\rho}}}\subset\Pscr_{\!(0,1)}$.
For the case of $p=\infty$, we have 
\begin{equation*}
	\|\vcg-\geps\|_{\infty}
	=
	\sup_{t\in\Zbb_+}{a}{{\rho}}^t
	=
	{a}<(1-{{\rho}})\varepsilon<\varepsilon.	
\end{equation*}
On the other hand, for $p\in[1,\infty)$, one can see that
\begin{equation*}
	\|\vcg-\geps\|_p 
	=
	{a}\Big(\sum_{t=0}^{\infty}{{\rho}}^{pt}\Big)^{\frac1p} 
	\!\!= 
	\frac{{a}}{(1-{{\rho}}^p)^{\frac{1}{p}}}
	<
	\frac{(1-{{\rho}})\varepsilon}{(1-{{\rho}}^p)^{\frac{1}{p}}}
	\le 
	\varepsilon,	
\end{equation*}
where the last inequality is due to  ${{\rho}}^p+(1-{{\rho}})^p\le 1$ which holds for any ${{\rho}}\in(0,1)$ and $p\in[1,\infty)$.
\qed
\subsection{Proof of Theorem \ref{thm:kernel_TC_DC_SS_are_FHR}} \label{sec:pf_kernel_TC_DC_SS_are_FHR}
Let $c_{00}$ be the space of impulse responses which are finitely non-zero, i.e., for each $\vcg = (g_s)_{s=0}^{\infty}$  there exists $\nG\in\Zbb_+$ such that $g_s=0$, for all $s\ge \nG$.
Note that for such $\vcg\in c_{00}$, we have $\Hankel(\vcg)\vcv\in\Rbb^{\nG}\times\{\zero\}$, for any $\vcv \in\ellinfty$.
This implies that $\rank\big(\Hankel(\vcg)\big)\le \nG<\infty$, and, therefore, $\vcg$ is a finite Hankel rank impulse response.
Let $\kernel:\Zbb_+\times\Zbb_+\to\Rbb$ be a finite support kernel, i.e., there exists $n_{\kernel}\in\Zbb_+$ such that $\kernel(s,t)=0$ when $s\ge n_{\kernel}$ or $t\ge n_{\kernel}$. One can easily see that $\kernel_t\in c_{00}$, for any $t\in\Zbb_+$. 
Therefore, we have $\rank\big(\Hankel(\kernel_t)\big)\le n_{\kernel}<\infty$, and consequently, $\kernel$ is a finite Hankel rank kernel.

Let $\vcf_{\beta}=(f_t)_{t=0}^\infty$ be an impulse response defined as $f_t=  \beta^t$, for $t\in\Zbb_+$. One can easily see that
\begin{equation}
	\{\Hankel(\vcf_{\beta})\vcv\ \! \big| \ \!  \vcv\in\ellinfty\Big\}
	=
	\{\vcf_{\beta}v\ \! \big| \ \! v\in\Rbb\}.
\end{equation}
Therefore, we have $\rank\big(\Hankel(\vcf_{\beta})\big)=1$, and consequently, $\vcf_{\beta}$ is a finite Hankel rank impulse response.
For the TC kernel introduced in \eqref{eqn:kernelTC} and $t\in\Zbb_+$, consider the section of $\kernelTC$ at $t$, i.e., $(\kernel_{\TC,t})_{s=0}^{\infty}$. Note that we have $\kernel_{\TC,t}=\vcf_{\beta}+\vcg$, where impulse response $\vcg=(g_s)_{s=0}^{\infty}$ is defined by $g_s= \beta^{\max(s,t)}-\beta^s$, for $s\in\Zbb_+$. 
One can easily see that $g_s=0$, for all $s\ge t$, which implies that $\vcg\in c_{00}$, and subsequently, $\rank\big(\Hankel(\vcg)\big)<\infty$.
Accordingly, since rank of $\Hankel(\vcf_{\beta})\big)$ is finite, it follows that 
$\rank\big(\Hankel(\kernel_{\TC,t})\big)<\infty$, and consequently, $\kernelTC$ is a finite Hankel rank kernel. 
Based on similar arguments, one can show same result for $\kernelDC$ and  $\kernelSS$, concluding the proof.
\qed
\bibliographystyle{IEEEtran}
\bibliography{mainbib_pos}
\end{document}